\documentclass{article}
\usepackage{amsmath,amssymb,amsthm}
\usepackage{centernot}
\usepackage{xcolor,cite}
\usepackage{graphicx}
\usepackage{subfigure}
\usepackage{multirow,enumerate}
\usepackage{fullpage}
\usepackage[titletoc,title]{appendix}

\usepackage{tikz}
\usetikzlibrary{arrows,shapes,trees,calc}

\def\PP{\mathbb P}

\def\RR{\mathbb R}

\def\cE{\mathcal E}

\def\b1{\mathbf 1}

\def\pi{{u}}

\newcommand{\hide}[1]{}

\newtheorem{theorem}{Theorem}

\newtheorem{proposition}{Proposition}

\newtheorem{claim}{Claim}
\newtheorem{remark}{Remark}
\newtheorem{conjecture}{Conjecture}
\newtheorem*{observation1}{Observation 1}
\newtheorem*{observation2a}{Observation 2A}
\newtheorem*{observation2b}{Observation 2B}
\newtheorem*{observation3a}{Observation 3A}
\newtheorem*{observation3b}{Observation 3B}

\title{A three-person deterministic graphical game\\ without Nash equilibria
\thanks{The first author thanks the National Science Foundation for partial support (Grant and IIS-1161476). The second author was partially funded by the Russian Academic Excellence Project `5-100'.
The work of the third author is supported in part by the Slovenian Research Agency (I$0$-$0035$, research program P$1$-$0285$, research projects N$1$-$0032$, J$1$-$5433$, J$1$-$6720$, J$1$-$6743$, and J$1$-$7051$). The work for this paper was done in the framework of
bilateral projects between Slovenia and the USA, partially financed by the Slovenian Research Agency (BI-US/$16$--$17$--$027$ and BI-US/$16$--$17$--$030$).}}

\author{
Endre Boros\thanks {MSIS \& RUTCOR, Business School, Rutgers University,
100 Rockafellar Road, Piscataway NJ  08854;
endre.boros@rutgers.edu}
\and
Vladimir Gurvich\thanks {MSIS \& RUTCOR,  Business School, Rutgers University,
100 Rockafellar Road, Piscataway NJ  08854 and
Research University: Higher School of Economics (HSE), Moscow, Russia; vladimir.gurvich@rutgers.edu}
\and
Martin Milani\v{c}\thanks {UP IAM and UP FAMNIT, University of Primorska, Koper, Slovenia; martin.milanic@upr.si}
\and
Vladimir Oudalov\thanks
{Salient Management Company, 15 Roszel Rd.~suite 109, Princeton NJ 08540-6248; oudalov@gmail.com}
\and
Jernej Vi\v{c}i\v{c}\thanks {UP IAM and UP FAMNIT, University of Primorska, Koper, Slovenia; jernej.vicic@upr.si}}

\date{\today}

\begin{document}

\maketitle

\begin{abstract}
We give an example of a three-person {deterministic graphical} game that has no Nash equilibrium in pure stationary strategies.
The game has seven positions, four outcomes (a unique cycle and three terminal positions), and
its normal form is of size $2 \times 2 \times 4$ only.
Thus, our example strengthens significantly the one obtained in 2014 by Gurvich and Oudalov;
the latter has four players, five terminals, and a $2\times 4 \times 6 \times 8$ normal form.
Furthermore, our example is minimal with respect to the number of players.
Both examples are tight but not Nash-solvable. Such examples were known since 1975, but
they were not related to {deterministic graphical} games.
Moreover, due to the small size of our example, we can strengthen it further by showing that
it has no Nash equilibrium not only in pure but also in independently mixed strategies, for both Markovian and a priori evaluations.

\medskip

\end{abstract}

\section{Introduction}
\label{s0}

{In 1950, Nash introduced his fundamental concept of equilibrium for \hbox{$n$-person} games \cite{Nas50, Nas51}.
In this paper we study the existence of a Nash equilibrium (NE) in pure and in independently mixed stationary strategies
in the so-called {\em deterministic graphical multi-person (DGMP) games}.
These are $n$-person positional games with perfect information, without positions of chance, and with terminal payoffs.
Such games are modeled by a finite directed graph that may have directed cycles.
The reader can find definitions and more details in Section \ref{s0a}.

There are two well-studied classes of DGMP games
for which the existence of an NE in pure stationary strategies is known:
two-person zero-sum games and acyclic games.
In the former case an NE is called a saddle point.
This case has a long history going back to pioneering works
by Zermelo~\cite{Zer1912}, K\"onig~\cite{Kon27}, and Kalm\'ar~\cite{Kal28}.
Further results were obtained in the much more general context of stochastic games,
for which Liggett and Lippman~\cite{LL69} corrected the original proof by Gillette~\cite{Gil57} and thus established the existence of a saddle point in pure stationary uniformly optimal strategies. The last property means that these strategies do not depend on the initial position; the corresponding equilibrium is called {\em subgame perfect}. The interested reader may find the necessary definitions and statements in~\cite{Was90, SW01, AHMS10,Gil57,LL69}.

After 1950, it became natural to ask whether we can extend these existence results to $n$-person games,
replacing the concept of a saddle point by the more general concept of an NE:

\begin{itemize}
\item[($Q$):] Does every DGMP game have an NE in pure stationary strategies?
\end{itemize}

The existence of a subgame perfect NE is out of question.
The corresponding examples were obtained already for the two-person case~\cite{AGH10}; see also \cite{BEGM11}.

\medskip

For acyclic games Kuhn~\cite{Kuh50, Kuh53} and Gale~\cite{Gal53} suggested the method of backward induction, which
answers question ($Q$) in the positive providing a subgame perfect NE in pure stationary strategies.
However this method is not applicable when the game contains cycles.
The interested reader can find more details in~\cite{AHMS10,Kuh50, Kuh53,Gal53,Mou83,Ewe01}.

Note, however, that the same position can appear several times, e.g., in chess.
In other words, the directed graph modeling a DGMP game may have (directed) cycles.
Thus, question ($Q$) remained open for DGMP games with cycles.
A positive answer was conjectured in~\cite[Conjecture 1]{BGMS07}, yet,
Gurvich and Oudalov disproved this conjecture in~\cite{GO14A, Gur15}
\footnote{One of the authors of~\cite{GO14A} refused to co-author paper~\cite{Gur15} published in Russia.},
thereby answering  ($Q$) in the negative.
Their example has $n=4$ players, $p=5$ terminals, a unique cycle, which is of length $5$, and
the normal form of size $2 \times 4 \times 6 \times 8$.

In this paper, we give a much smaller example: the game has only seven positions, $n = p = 3$,
there is a unique cycle, which is of length $2$, and the normal form of size $2 \times 2 \times 4$.
The small size of this example allows us to strengthen the result further by proving that
the obtained game has no NE not only in pure but also in independently mixed strategies,
for both a priori and Markovian evaluations.

\medskip

These two examples represent also an important strengthening of some other examples obtained long time ago.
It easily follows from results of~\cite[Section 3]{BG03} that game forms corresponding to DGMP games
have the so-called {\em tightness property}, which is closely related to Nash-solvability; see Section~\ref{s3} for more details.
For the two-person case they are just equivalent;
for three players, however, tightness does not imply Nash-solvability~\cite{Gur75,Gur88}
(nor vice versa~\cite{Gur88}). Yet, all these old examples do not arise from DGMP games, while the two new examples do.

Results of~\cite{Gur75,BG03} imply that our $3$-player example is minimal with respect to the number of players.
In other words, for any two-person DGMP game an NE in pure stationary strategies
exists not only in the zero-sum case but in general too.
(However, the property of subgame perfectness can only be added in the first case.)

\medskip

The paper is structured as follows. In Section~\ref{s0a} we introduce the basic concepts, definitions, and notation.
In Section~\ref{example} we give our main example, a three-person DGMP game with no NE in pure stationary strategies (Proposition~\ref{prop:no-NE-in-pure-strategies}).
We strengthen this result further in Section~\ref{independently-mixed}
by introducing payoffs such that the obtained games have no NE not only in pure but
also in independently mixed strategies, for both a priori and Markovian evaluations.
The proofs of these two claims are given in Appendix.
In Section~\ref{s3}, we show that our example is minimal with respect to the number of players and
compare it with tight but not Nash-solvable game forms from \cite{Gur75,Gur88}.
We conclude the paper with some open ends and conjectures in Section~\ref{conclusion}.}

\section{Basic definitions}
\label{s0a}

{ In~1990, Washburn~\cite{Was90} defined a {\em deterministic graphical (DG) game}
as a two-person zero-sum positional game with perfect information, without positions of chance, and with terminal payoffs.
We extend this concept and introduce the notion of a {\em deterministic graphical multi-person (DGMP) game}, which
may be not zero-sum, even in the case of two players.}
Such a game is modeled by a {\em finite directed graph} (digraph)  $G = (V,E)$  whose vertices are partitioned into
$n+1$  subsets, $V = V_1 \cup \ldots \cup V_n \cup V_T$. A vertex  $v \in V_i$  is interpreted
as a {\em position} controlled by player  $i \in I = \{1, \ldots, n\}$, while
a vertex  $v \in V_T = \{a_1, \ldots, a_p\}$  is a {\em terminal position}
(or just a {\em terminal}, for short), which has no outgoing edges.
A directed edge  $(v, v')$  is interpreted as a (legal) {\em move} from position  $v$  to  $v'$.
We also fix an initial position $v_0 \in V \setminus V_T$.

A {\em pure stationary strategy} of a player  $i \in I$  is a mapping that assigns a move $(v, v')$  to each position  $v \in V_i$.
This concept was introduced by Kalm\'ar~\cite{Kal28} under the name `tactic in the strict sense' (see~\cite{SW01}).
In the literature on stochastic games and Markovian decision processes such strategies are called `pure stationary strategies', see, e.g.,~\cite{MO70}. The term {\em pure positional strategies} is also common in the literature~\cite{BG09,AHMS10,BEGM11}.\footnote{In almost all sections of this paper, we restrict ourselves and the players to their pure stationary strategies.
The history dependent strategies will not be considered at all, while the mixed ones will only be considered in Section~\ref{independently-mixed}. In other sections words ``pure stationary'' may be given in parentheses or omitted.}
Once each player $i \in I$ chooses a (pure stationary) strategy  $s^i$, we obtain a collection
$s = (s^i \mid i \in I)$  of strategies, called a {\em strategy profile}.
In a DGMP game modeled by a digraph, to define a strategy profile
we have to choose, for each vertex  $v \in V \setminus V_T$, a single outgoing arc $(v, v') \in E$.
A strategy profile $s$ can be also viewed as a subset of the arcs, $s \subseteq E$.
Clearly, in the subgraph $G_s = (V,s)$ corresponding to a strategy profile  $s$  there is a unique walk that
begins at the initial vertex  $v_0$  and either ends in a terminal or results in a cycle.
In other words,  $s$  uniquely defines a {\em play}  $P(s)$  that begins in  $v_0$  and
either ends in a terminal $a \in V_T$  or forms a ``lasso'', that is, consists of a
(possibly empty) initial path and a cycle repeated infinitely.
This follows from the assumption that  $s$  consists of stationary strategies.

{We also assume that all cycles of  $G$  are equivalent, that is,
they form a unique outcome  $c$  of the game.
Thus, the game has  $q = p+1$  outcomes that form the set  $A = \{a_1, \ldots, a_p, c\}$.
This is in agreement with Washburn's~\cite{Was90} assumption that
all infinite plays result in zero payoff for every player.}

\medskip
To each pair consisting of a player  $i \in I$  and outcome $a \in A$
we assign a {\em payoff} value $u(i, a)$ specifying the profit of player  $i$ if outcome $a$ is realized.
The corresponding mapping  $u : I \times A \to \RR$  is called the {\em payoff function}.
We will write  $u_i(a)$ instead of $u(i, a)$.

A {\em DGMP game in positional form} is a quadruple  $(G, D, u, v_0)$, where $G = (V, E)$  is a digraph,
$D$  is a partition $V = V_1 \cup \ldots \cup V_n \cup V_T$ of the set of positions,
$u: I \times A \to \RR$  is a payoff function, and $v_0 \in V$  is a fixed initial position.
The triplet  $(G, D, v_0)$  is called a {\em DGMP game structure}.

\begin{remark}
Standardly, DGMP games studied in this paper can also be represented by trees.
However, in the presence of cycles such trees are infinite.
Thus, using the tree representation would not be convenient.
\end{remark}

For a vertex  $v \in V$ let us denote by $d^+_v$ its out-degree.
Then, player  $i$  has $\prod_{v\in V_i}d^+_v$  many different strategies. We denote this set of strategies by
$$S_i=\left\{s^i_j : j=1,\ldots,\prod_{v\in V_i}d^+_v\right\}\,.$$
This way we obtain a {\em game form}, that is, a mapping  $g : S \rightarrow A$, where $S = S_1 \times  \ldots \times S_n$
is the Cartesian product of the sets of strategies of players  $i \in I$.
By this definition,  $g(s) = c$  when  $P(s)$  is an infinite play and  $g(s) = a$  when  $P(s)$  terminates in  $a \in V_T$.
The {\em normal form} of a  DGMP game $(G, D, u, v_0)$ is defined as the pair $(g, u)$, where
$g$ is the game form of the DGMP game structure $(G, D, v_0)$ and $u: I \times A \to \RR$ is a payoff function.
Let us note that game forms and games in normal form are of independent interest
and the former can be studied without underlying DGMP game structures.
{It is important to understand that not every game form is a normal form of a DGMP game structures;
for example, the latter must be tight; see Section 5 for the definition.}

\smallskip
Given a game  $(g, u)$ in normal form, let us fix a strategy profile
$s = (s^i \; | \; i \in I) \in S$  and a player $i_0 \in I$.
Then, we can represent  $s$  as  $(s^{i_0}, s^{-i_0})$, where
$s^{-i_0}$  is the collection of strategies in $s$ of the other $n-1$ players $i \in I \setminus \{i_0\}$.
We say that player  $i_0$  {\em can improve on $s$}  if (s)he has a strategy  $\bar{s}^{i_0}$  such that
$$u_{i_0}(g((\bar{s}^{i_0}, s^{-i_0}))) > u_{i_0}(g((s^{i_0}, s^{-i_0}))).$$
A strategy profile  $s$  is called a \emph{Nash equilibrium} (an NE) if no player can improve on $s$.
In other words,  $s$  is an NE if no player $i \in I$  can profit by replacing his/her strategy  $s^i$  in  $s$
by a new strategy  $\bar{s}^i$, assuming that the other $n-1$ players keep their strategies from  $s$  unchanged.
A game form $g$ is called {\em Nash-solvable} (NS) if for every payoff function
$u:I \times A \to \RR$ the corresponding game $(g,u)$ has an NE.

Note that shifting the payoff of any player by a constant does not change the set of NEs of the game.
In particular, we may assume without loss of generality that  $\min_{a \in A}u_i(a) = 0$  for all players $i\in I$.
Furthermore, let us also note that, by the above definition, one does not need ties, $u_i(a) = u_i(a')$, to construct NE-free examples.

\section{Main example}
\label{example}

We now describe our main example:
a three-person DGMP game $(G, D, u, v_0)$ in positional form without an NE in pure stationary strategies.
Figure~\ref{fig0}(a) shows the corresponding DGMP game structure $(G, D, v_0)$, where
the initial position $v_0$ is indicated by the square.
The set of outcomes is  $A = \{a_1,a_2,a_3,c\}$, where, as usual, $c$ represents the cyclic outcome
(the unique infinite play).
The payoff function can be any function $u:I\times A\to \mathbb{R}$ satisfying the inequalities
listed in Figure~\ref{fig0}(b). The game in normal form is given in Figure~\ref{fig0}(c).
It is constructed from the positional form in accordance with the definitions of Section~\ref{s0a}.

\begin{figure}[h!]
\centering
\includegraphics[width=0.92\textwidth]{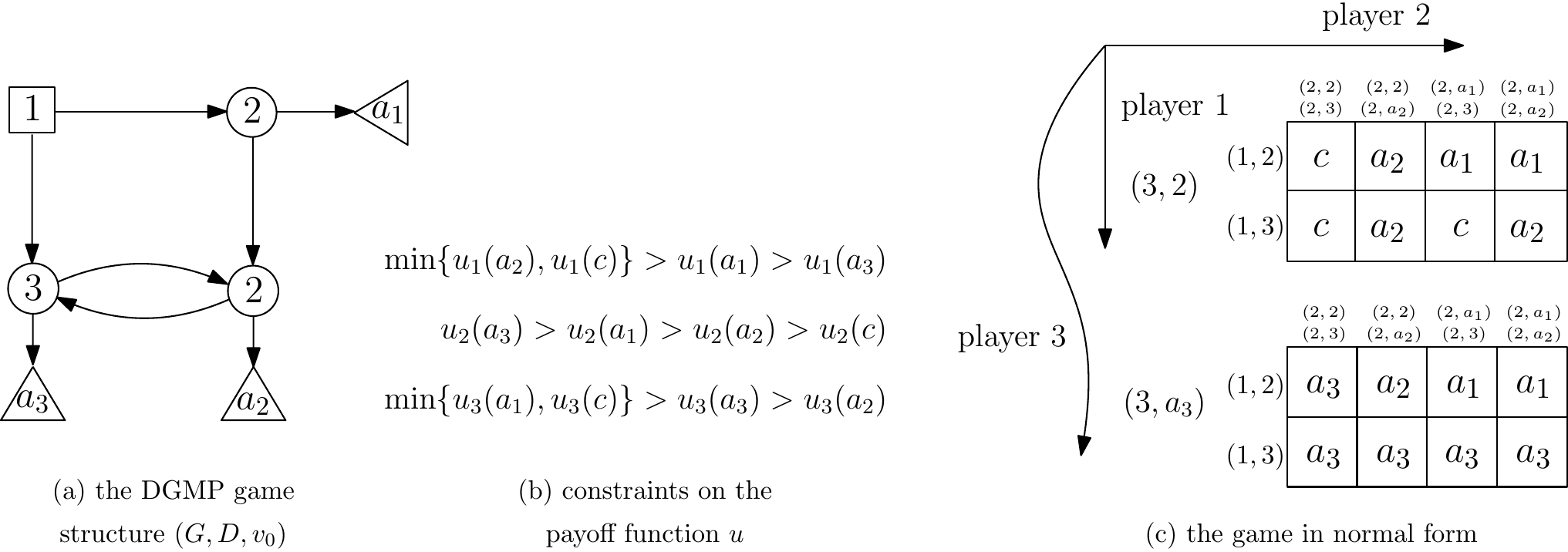}
\caption{Our main example: an NE-free three-person DGMP game in positional and normal forms.}\label{fig0}
\end{figure}

\begin{proposition}\label{prop:no-NE-in-pure-strategies}
The three-person DGMP game in Figure~\ref{fig0} has no NE in pure stationary strategies.
\end{proposition}

\begin{proof}
It is sufficient to show that for every $s \in S$ there exists a player $i \in I$  and
a strategy profile  $\bar{s}$ that differs from  $s$ only in the coordinate  $i$  such that $u_i(g(\bar{s}))>  u_i(g(s))$.

\begin{figure}[h!]
\centering
\includegraphics[width=0.4\textwidth]{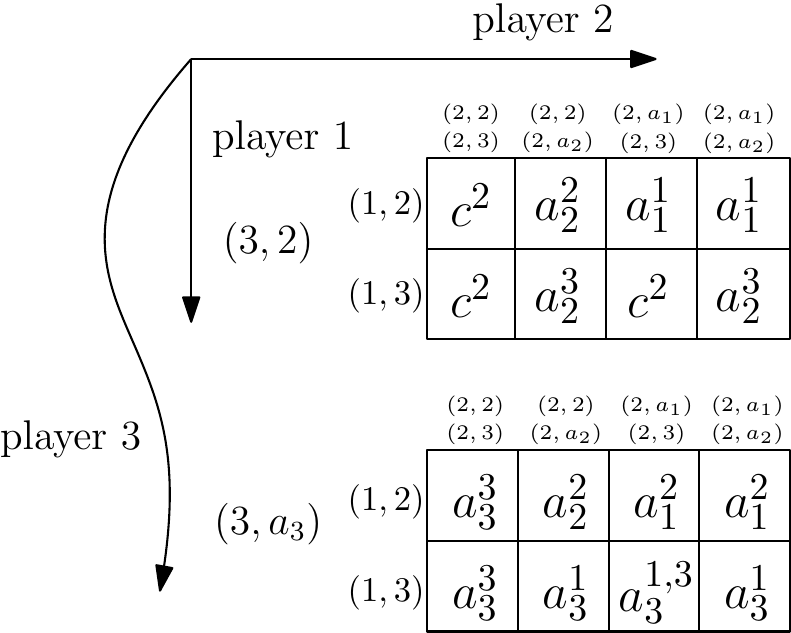}
\caption{The three-person game depicted in Figure~\ref{fig0} has no NE in pure stationary strategies.
To verify this, we only need to check that
from each entry of the table, it is possible to move in one of the three directions,
say, in direction $i$, to another entry specifying an outcome that is strictly better than the original one for player $i$.
All such directions (players) are indicated by the upper indices, for each entry.}
\label{fig-no-NE-in-pure}
\end{figure}

We can easily verify this using Figure~\ref{fig-no-NE-in-pure}; see the caption.
For example, the entry corresponding to $s = ((1,2), ((2,a_1),(2,a_2)), (3,a_3))$
is $a_1^2$, indicating that there exists a strategy profile $\bar{s}$
that differs from $s$ only in coordinate $2$ and such that  $u_2(g(\bar{s})) > u_2(g(s))= u_2(a_1)$.
Indeed, if player $2$ replaces the current strategy $((2,a_1), (2,a_2))$ with $((2,2),(2,3))$,
the obtained strategy profile $\bar{s} =  ((1,2),((2,2),(2,3)),(3,a_3))$ will result in an improvement,
since $u_2(g(\bar{s})) = u_2(a_3) > u_2(a_1) = u_2(g(s))$.
Note that the outcome $g(s) = a_3$  of the strategy profile
$s = ((1,3), ((2,a_1),(2,3)), (3,a_3))$  has two upper indices, because
$s$ can be improved by each of the two players $1$ and~$3$.
\end{proof}

\begin{remark}
The second and the third strategies of player $2$ are dominated by the last one.
However, in general, a dominated strategy could participate in a (pure or mixed) NE.
Thus, we cannot simply discard those two strategies of player $2$.
\end{remark}

\begin{remark}
The above example contains a unique cycle, which is of length $2$.
The game restricted to the positions of this cycle and terminals $a_2$ and $a_3$
has no subgame perfect NE; see~\cite{AGH10, BEGM11}. This is not a coincidence: we remark
(without giving a proof) that in any NE-free example with a unique cycle,
the restriction of the game to the positions of this cycle and all positions that can be reached from it
(not  necessarily in one move), results in a game with no subgame perfect NE.
This idea was the starting point for our computer based search of NE-free games.
\end{remark}

\section{NE-freeness in independently mixed strategies}
\label{independently-mixed}

In Section~\ref{example} we considered pure strategies; now we turn to {\it mixed} ones.
Mixed strategies of a player are standardly defined as probability distributions over the
(finite) set of all pure strategies of this player.
An {\em independently mixed} strategy of player $i\in I$ { is a collection of probability distributions,
one for each position $v \in V_i$, on the set of moves $(v,v')$ from $v$.}
Obviously, every independently mixed strategy of player $i$ can be viewed as a mixed strategy, but
not vice versa unless player  $i$  controls a unique position.

The notion of an NE is extended in the natural way, by taking into account the expected payoff of the players; see Appendix~A.

\medskip
In a DGMP game, a play may visit the same position more than once, and
we distinguish the following two evaluations of the game played
in fixed independently mixed strategies.
The {\em Markovian evaluation} assumes that for every player  $i$  and for each position of this player,
(s)he randomizes the choice of the move each time the play comes to this position.
In contrast, the {\em a priori evaluation} assumes that for every player $i$ and for each position of this player,
the random choice of the move at this position is done only once, before the play begins, and
the obtained move is used whenever the play comes to this position again.
Clearly, the probabilities $\PP(s,a)$ of an outcome $a$  might differ for these two models.
The concept of a priori evaluation was introduced in 2013 by Boros et al.~\cite{BGY13}.

We now state two theorems explaining how the result of Proposition~\ref{prop:no-NE-in-pure-strategies}
can be strengthened to examples having no NE not only in pure but also in independently mixed strategies,
for both Markovian and a priori evaluations.

In the case of Markovian evaluation, no modification to the example from Figure~\ref{fig0} is needed.

\begin{theorem}\label{thm:Markovian}
The three-player DGMP game described in Figure~\ref{fig0} has no NE in independently mixed strategies
for the Markovian evaluation.
\end{theorem}

A proof of Theorem~\ref{thm:Markovian} is given in  Appendix~B.

For the a priori evaluation, the example from Figure~\ref{fig0} is modified by strengthening the set of
constraints on the payoff function $u$ with a set of more specific constraints.

\begin{theorem}\label{thm:a-priori}
The three-person DGMP game described in Figure~\ref{fig0} has no NE in independently mixed strategies
for the a priori evaluation whenever the payoff function $u$ satisfies the following conditions:
\begin{eqnarray}
u_1(a_2) > u_1(c) >u_1(a_1) >u_1(a_3)= 0\,,\\
u_2(a_3) > u_2(a_1) >u_2(a_2)>u_2(c) = 0\,,\\
{\min}\{u_3(a_1), u_3(c)\} >u_3(a_3)>u_3(a_2)= 0\,,
\end{eqnarray}
\begin{eqnarray}
\label{assumption1}
\frac{u_2(a_3)-u_2(a_1)}{u_2(a_3)-u_2(a_2)}&<& 1/2\,,\\
\label{assumption2}
\frac{u_3(a_3)}{u_3(c)}&<& 1/2\,,\\
\label{assumption3}
\frac{u_1(a_1)}{u_1(c)}+\frac{u_2(a_1)}{u_2(a_3)}&<&1\,.
\end{eqnarray}
\end{theorem}

A proof of Theorem~\ref{thm:a-priori} is given in Appendix~C.

\medskip
For a specific payoff function satisfying the assumptions of Theorem~\ref{thm:a-priori}
let us consider for instance the following function:
\begin{equation*}
\begin{aligned}
u_1(a_2) = 4 > u_1(c) = 3   > u_1(a_1) = 1>u_1(a_3) = 0\,,\\
u_2(a_3) = 8 > u_2(a_1) = 5 > u_2(a_2) = 1>u_2(c) = 0\,,\\
u_3(a_1) = 4 > u_3(c) = 3   > u_3(a_3) = 1>u_3(a_2) = 0\,.
\end{aligned}
\end{equation*}
It is straightforward to verify that the given payoff function satisfies the requirements.

\section{Nash-solvability versus tightness}
\label{s3}

Let us look at our main example from a different point of view. In particular, we will explain why it is minimal with respect to the number of players and discuss an important feature present in this example, along with the earlier example from~\cite{GO14A, Gur15}, that
was not present in older examples of tight (the definition is given below) but not Nash-solvable game forms.

\subsection{The two-person case}\label{sec:2-person-case}

The next theorem shows that the number $n=3$ of players cannot be reduced.

\begin{theorem}
\label{tvg1}
Every two-person DGMP game has an NE in pure stationary strategies.
\end{theorem}

A proof of this result can be found in~\cite[Section $3$]{BG03}; see also~\cite[p.~1485]{BGMS07}.
We sketch the proof here in order to explain the basic relations between Nash solvability.

A two-person game form  $g$  is called:
\begin{enumerate}[(i)]
\item {\em Nash-solvable} (NS) if for every payoff function  $u=(u_1,u_2)$ the game $(g,u)$ has an NE
(in agreement with the definition of Nash-solvability of $n$-person games given in Section~\ref{s0a}).
\item {\em zero-sum-solvable} if for every
$u=(u_1,u_2)$ that satisfies $u_1(a) + u_2(a) = 0$  for all $a\in A$
the game $(g,u)$ has an NE.
\item  {\em $\pm 1$-solvable}  if
the corresponding game  $(g,u)$  has an NE for every
$u = (u_1, u_2)$  such that
$u_1(a) + u_2(a) = 0$  for each outcome  $a \in A$
and both  $u_1$  and  $u_2$  take only values  $+1$  or  $-1$.
\end{enumerate}

Recall that an NE of a two-person zero-sum game,
as in cases (ii) and (iii), is called a saddle point.
\begin{sloppypar}
\begin{proposition}
\label{lvg1}
The above three properties of two-person game forms,
(i) Nash-solvability,
(ii) zero-sum-solvability, and
\hbox{(iii) $\pm 1$-solvability}, are equivalent.
\qed
\end{proposition}
\end{sloppypar}

The equivalence of (ii) and (iii) was shown in 1970 by Edmonds and Fulkerson \cite{EF70}, and independently in \cite{Gur73}.
Later (i) was added to the list in \cite{Gur75, Gur88}, together with one more important property called tightness.
We recall its definition for the general, $n$-person, case in the next six paragraphs.

\smallskip

\noindent{\bf Coalitions, blocks, effectivity function, and tightness.}
Let $g : S = \prod_{i \in I} S^i  \rightarrow A$  be an $n$-person game form.
In the game (voting) theory, $I$ and $A$ standardly denote the sets
of players (voters) and outcomes (candidates), respectively.
We assume that mapping  $g$  is surjective, $g(S) = A$, but typically it is not injective,
in other words, different strategy profiles may determine the same outcome.
The subsets $K \subseteq I$  and  $B \subseteq A$
are called {\em coalitions} and {\em blocks}, respectively.
A strategy  $s^K_{} = (s^i \; | \; i \in K)$  of a coalition  $K$ is
defined as a collection of strategies, one for each coalitionist of $K$.
A pair of strategies  $s^K_{}$  of coalition $K$ and
$s^{I \setminus K}_{}$ of the complementary coalition $I \setminus K$
uniquely defines a strategy profile  $s = (s^K_{}, s^{I \setminus K}_{})$.

We say that a non-empty coalition  $K$  is {\em effective} for a block  $B$
if this coalition has a strategy such that
the resulting outcome is in  $B$  for any strategies of the opponents;
in other words, if there is a strategy  $s^K_{}$
such that for any strategy  $s^{I \setminus K}_{}$
for the obtained strategy profile  $s = (s^K_{}, s^{I \setminus K}_{})$  we  have  $g(s) \in B$.
We write  $\cE(K,B) = 1$  if  $K$  is effective for  $B$  and  $\cE(K,B) = 0$  otherwise.

Since  $g$  is surjective, by the above definition, we have
$\cE_g(I, B)= 0$  if and only  $B = \emptyset$.
The value  $\cE_g(K,B)$  was not yet defined for $K = \emptyset$.
By convention, we set  $\cE_g(\emptyset, B)= 1$  if and only if  $B = A$.

The obtained function   $\cE_g: 2^I \times 2^A = 2^{I \cup A} \rightarrow \{0,1\}$
is  called the {\em effectivity function} (of the game form $g$).
By definition, $\cE_g$ is a Boolean function whose set of variables
is the union $I \cup A$ of all players and outcomes.

Note that equalities $\cE_g(K,B) = 1$ and  $\cE_g(I \setminus K, A \setminus B)= 1$  cannot hold simultaneously.
Indeed, otherwise there would exist strategies $s^K_{}, s^{I \setminus K}_{}$ and blocks  $B,  A \setminus B$  such that
$g(s) \in B \cap (A \setminus B) = \emptyset$, that is,
$g$  is not defined on the strategy profile  $s = (s^K_{}, s^{I \setminus K}_{})$.
Note that cases  $K = \emptyset$  and  $K = I$  are covered by the above convention.
Thus, $\cE_g(K,B) = 1 \Rightarrow \cE_g(I \setminus K, A \setminus B) = 0$.

A game form  $g$  is called {\em tight} if the inverse implication holds too:
$\cE_g(I \setminus K, A \setminus B)= 0 \Rightarrow \cE_g(K,B) = 1$, or equivalently, if
$\cE_g(K,B) = 1 \Leftrightarrow \cE_g(I \setminus K, A \setminus B)= 0$
for all pairs  $K \subseteq I$  and $B \subseteq A$, or,
in other words, if  $\cE_g$  is a self-dual Boolean function, see, e.g.,~\cite{CH11} for the definition.

\medskip

It is not difficult to see~\cite{Gur75,Gur88} that for the two-person game forms,
tightness is equivalent with  $\pm 1$-solvability and, thus, can be added to Proposition~\ref{lvg1}.

\begin{sloppypar}
\begin{theorem}\label{four-properties}
The following properties of two-person game forms:
(i) Nash-solvability, (ii) zero-sum-solvability,
\hbox{(iii) $\pm 1$-solvability}, and \hbox{(iv) tightness},
are equivalent.
\end{theorem}
\end{sloppypar}

The following theorem completes the proof sketch of Theorem~\ref{tvg1}.

\begin{theorem}\label{chess-like-tight}
Every game form corresponding to a DGMP game structure is tight.
\end{theorem}

\begin{proof}[Proof sketch.]
Let $g : S = \prod_{i \in I} S^i  \rightarrow A$  be a game form corresponding to a DGMP game structure $(G,D,v_0)$.
To verify that  $g$  is tight, we have to check the implication
$\cE_g(I \setminus K, A \setminus B)= 0 \Rightarrow \cE_g(K,B) = 1$
for all pairs  $K \subseteq I$ and $B \subseteq A$.
But  this can be easily reduced to the two-person case,
simply by replacing coalitions $K$ and $I\setminus K$ by players $1$ and $2$, respectively.
Then, we have to verify tightness of the corresponding two-person game forms for every unordered pair $\{K, I\setminus K\}$.
Tightness follows from the above convention whenever coalition  $K$  or  $I \setminus K$  is empty;
hance, we can assume that  $\emptyset \neq K \neq I$.
By Theorem~\ref{four-properties}, tightness for two-person game forms is equivalent with $\pm 1$-solvability.
The fact that every DGMP game form is $\pm 1$-solvable was proved by Boros and Gurvich~\cite[Section 3]{BG03};
see also~\cite[Section 12, p.~1485]{BGMS07}.
\end{proof}

\begin{remark}
It is essential for the above arguments that all cycles form a unique outcome $c$ of the game. For another natural case, when each cycle is a separate outcome, a criterion of Nash-solvability was obtained in~\cite{BGMS07}, for the two-person case.
\end{remark}

\subsection{The $n$-person case for $n>2$}\label{s5}

Although the concepts of tightness and Nash-solvability are defined for any number $n$ of player,
for $n > 2$ they are no longer equivalent. In particular, unlike tightness (see the above proof sketch of Theorem~\ref{chess-like-tight}),
Nash-solvability of $n$-person game forms cannot be reduced to the two-person case.
Moreover, tightness is neither necessary nor sufficient for Nash-solvability.
An example of a tight but not Nash-solvable game form
was first given in \cite{Gur75}; see Remark 3 therein.
This example has the same number of players, $n=3$, and
the same size,  $2 \times 2 \times 4$, as our main example, but it has only three outcomes.
A smallest possible such example, of size $2 \times 2 \times 2$, was suggested by Andrew Volgin in $2012$
(at the time he was a high-school student); see Figure~\ref{Volgin}. However, this example has four outcomes.

\begin{figure}[h!]
\centering
\includegraphics[width=0.75\textwidth]{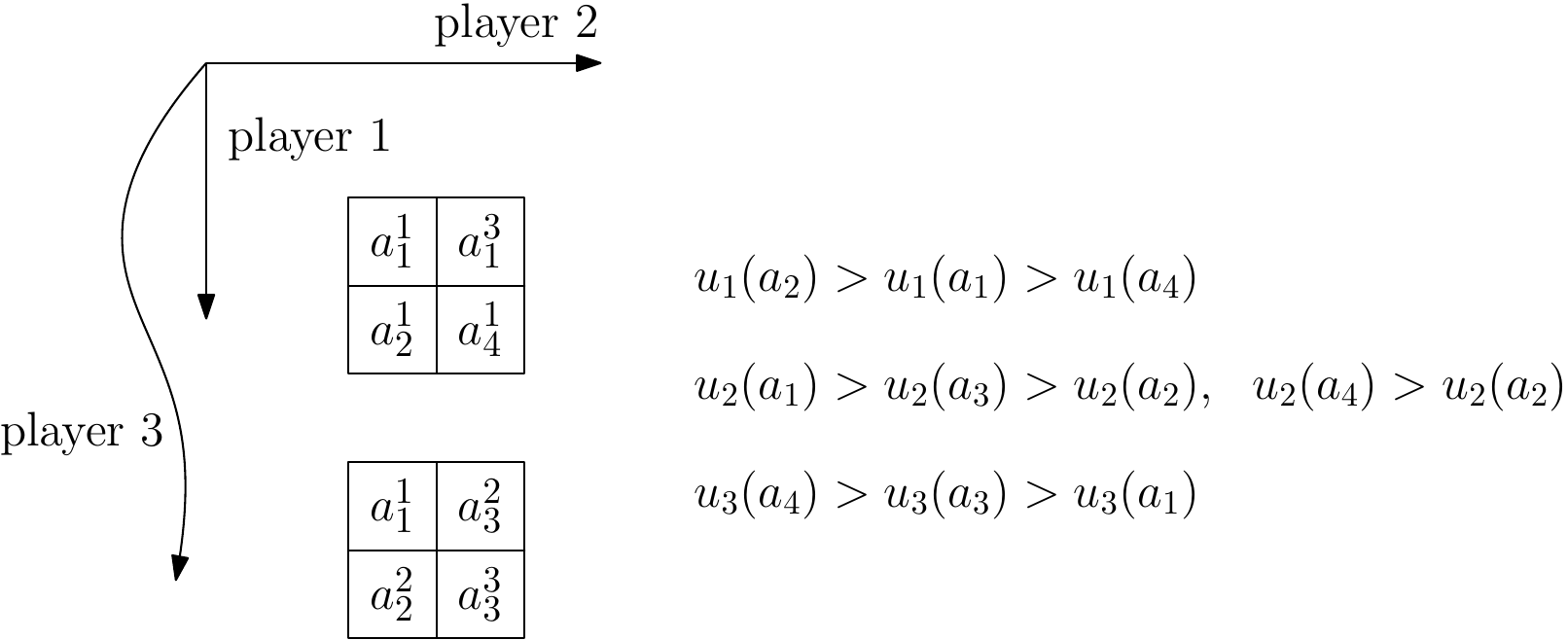}
\caption{A $2 \times 2 \times 2$ game form that is tight but not Nash-solvable.
We use the same notation as in Figures~\ref{fig0} and~\ref{fig-no-NE-in-pure}.}
\label{Volgin}
\end{figure}

Another, fully symmetric example with three players and three outcomes was given in  \cite{Gur88} but
it is of larger size, $6 \times 6 \times 6$; it is reproduced also in \cite{BG03}.

\smallskip

Let us emphasize that none of the above examples represents the normal form of a DGMP game structure,
while the three-person  $2 \times 4 \times 2$  game form constructed in the present paper does;
furthermore, it is also tight but not Nash-solvable (by Theorem~\ref{chess-like-tight} and
Proposition~\ref{prop:no-NE-in-pure-strategies}, respectively). The same properties hold also for the larger example due to Gurvich and Oudalov~\cite{GO14A, Gur15}.

\section{Open ends}\label{conclusion}

\subsection{The case when $c$ is the worst outcome}\label{c-is-the-worst}

By definition, all cycles of the digraph modeling a DGMP game form a unique outcome denoted by $c$, hence,
such a game has $p+1$  outcomes that form the set $A = \{a_1, \ldots, a_p, c\}$.
Let us consider the following assumption and the corresponding reformulation of question ($Q$):
\begin{itemize}
%\noindent
\item[($C$):] For every player, outcome $c$ is strictly worse than any of the terminals.
\item[($Q_C$):] Does every DGMP game satisfying condition ($C$) have an NE in pure stationary strategies?
\end{itemize}

While question ($Q$) was answered in the negative, question $(Q_C)$ remains open.\footnote{Our main example and the example from \cite{GO14A, Gur15} do not satisfy $(C)$.} In the last section of \cite{BG03} it was shown that the answer to ($Q_C$) is positive for $p=2$.
An attempt to extend this result to $p=3$ was made in \cite{BR09}.
However, later the authors found a flaw in the proof and, thus, this case remains open,
while question ($Q$) remains open for $p=2$.

It was conjectured in \cite[Conjectures 1 and 2]{BGMS07} that ($Q$) and $(Q_C)$ have a positive answer.
Later, in~\cite{GO14}, the first conjecture was disproved and
it was conjectured that second question have a negative answer. We second this conjecture.

\begin{conjecture}\label{c0}
There exists a DGMP game satisfying condition ($C$) and having no NE in pure stationary strategies.
\end{conjecture}

Some supporting arguments for this conjecture can be derived from results of~\cite{BEGM11}, as follows.
In Table 1 of that paper, several concepts of solvability are summarized.
Consider the first four rows of this table. In even rows condition ($C$) is assumed, while in odd rows it is not.
Thus, this part of the table contains $16$ pairs of questions
(similar to $(Q)$ and $(Q_C)$) and for each pair, condition ($C$) is required in one of the two questions but not in the other one.
In \cite{BEGM11}, $30$ questions ($15$ pairs) are answered.
Among these $30$ answers, $22$  are negative and $8$ are positive, but interestingly, no answer ever depends on condition ($C$),
that is, in each of these $15$ pairs either both answers are negative, or both are positive;
in other words, condition ($C$) turns out to be irrelevant in all these cases.

Our pair of questions, ($Q$) and ($Q_C$), form the $16$th pair of the table.
Our main example along with the example from~\cite{GO14A, Gur15} answers ($Q$) in the negative, while ($Q_C$) remains open.
However, taking the above observations into account, we second \cite{BEGM11, GO14} and conjecture that
the answer to ($Q_C$) is the same as to ($Q$), that is, negative.
However, the counterexample for ($Q_C$) may need to be much larger than our example for~($Q$).

Note also that a corresponding example, if there is one, would strengthen simultaneously
the main example of the present paper and of \cite{GO14}; see Figure~$2$ and Table~$2$ therein.

\bigskip
Furthermore, we now explain that a positive answer to Conjecture~\ref{c0} would imply the following curious statement:
there exists a DGMP game satisfying ($C$) in which $c$ is the only NE outcome (in spite of the fact that $c$ is worse than any terminal for each player). In an arbitrary DGMP game a strategy profile
{(in particular an NE)} is called {\em terminal} if the corresponding outcome is a terminal position and {\em cyclic} if it is  $c$. A DGMP game is {\em connected} if every position is reachable from the initial one.

\begin{proposition}
\label{p4}
If the answer to question ($Q_C$) is negative, then there exists a connected DGMP game that
satisfies ($C$) and has a (cyclic) NE  but no terminal one, although the game has at least one terminal.
\end{proposition}

The last property is important to rule out trivial terminal-free examples such as a directed cycle.

\begin{proof}
Take any connected DGMP game $\Gamma'$ that has a cyclic NE  $s'$,
for example, one of the two games shown in Figure \ref{fig1}.

\begin{figure}[h!]
\begin{center}
\begin{tikzpicture}[scale=0.5]
\path
(3,3) node(v0)[draw,black,circle,fill=red!50!white] {$\mathbf{1}$}
(6,0) node(v1)[draw,black,circle] {$\mathbf{1}$}
(0,0) node(v2)[draw,black,circle] {$\mathbf{2}$}
(6,3) node(v3)[draw,black,circle] {$\mathbf{2}$}
(6,6) node(v4)[draw,black,circle] {$\mathbf{3}$}
(3,0) node(v5)[draw,black,circle] {$\mathbf{3}$}
(0,3) node(t2)[draw,black,rectangle] {$\mathbf{a_2}$}
(3,6) node(t3)[draw,black,rectangle] {$\mathbf{a_3}$}
(3,-3) node(t1)[draw,black,rectangle] {$\mathbf{a_1}$};
\foreach \x/\y in {v0/v3,v0/v2,v1/v3,v1/t1,v2/v5,v2/t2,v3/v5,v3/v4,v4/v0,v4/t3,v5/v0,v5/v1}
\draw[color=black!80!white,thick,->] (\x) -- (\y);
\foreach \x/\y in {v0/v3,v3/v5,v5/v0} \draw[color=red,line width=2pt,->] (\x) -- (\y);
\foreach \x/\y in {v4/v0,v2/v5,v1/v3} \draw[color=blue,line width=2pt,->] (\x) -- (\y);
\end{tikzpicture}
\hspace*{10mm}
\begin{tikzpicture}[scale=0.5]
\path
(0,0) node(v0)[draw,black,circle,fill=red!50!white] {$\mathbf{1}$}
(-3,3) node(v1)[draw,black,circle] {$\mathbf{2}$}
(0,3) node(v2)[draw,black,circle] {$\mathbf{3}$}
(3,3) node(v3)[draw,black,circle] {$\mathbf{4}$}
(-4.5,6) node(v4)[draw,black,circle] {$\mathbf{3}$}
(-1.5,6) node(v5)[draw,black,circle] {$\mathbf{4}$}
(1.5,6) node(v6)[draw,black,circle] {$\mathbf{2}$}
(4.5,6) node(v7)[draw,black,circle] {$\mathbf{1}$}
(-3,9) node(t1)[draw,black,rectangle] {$\mathbf{a_1}$}
(0,9) node(t2)[draw,black,rectangle] {$\mathbf{a_2}$}
(3,9) node(t3)[draw,black,rectangle] {$\mathbf{a_3}$};
\foreach \x/\y in {v0/v2,v0/v3,v1/v0,v1/v4,v2/v1,v2/v5,v2/v6,v3/v2,v3/v6,v3/v7,v4/v5,v4/t1,v5/v1,v5/t1,v5/t2,v5/v6,v6/t2,v6/t3,v6/v7,v7/t3,v7/v2}
\draw[color=black!80!white,thick,->] (\x) -- (\y);
\foreach \x/\y in {v0/v3,v3/v2,v2/v1,v1/v0,v4/v5,v5/v1,v6/v7,v7/v2} \draw[color=blue,line width=2pt,->] (\x) -- (\y);
\foreach \x/\y in {v0/v3,v3/v2,v2/v1,v1/v0} \draw[color=red,line width=2pt,->] (\x) -- (\y);
\end{tikzpicture}
\end{center}
\caption{Two examples of DGMP games with three terminals $V_T'=\{a_1,a_2,a_3\}$.
Players are indicated inside the nodes. In the left we have three players, while in the right example there are four.
The shaded nodes controlled by player $1$ are the initial vertices.
In both cases, thick lines form strategy profiles resulting in a cycle.
It is easy to verify that the pictured strategy profiles are NE in both cases,
regardless of the players' preferences over the outcomes, that is, even if they all rank $c$ the last.}
\label{fig1}
\end{figure}
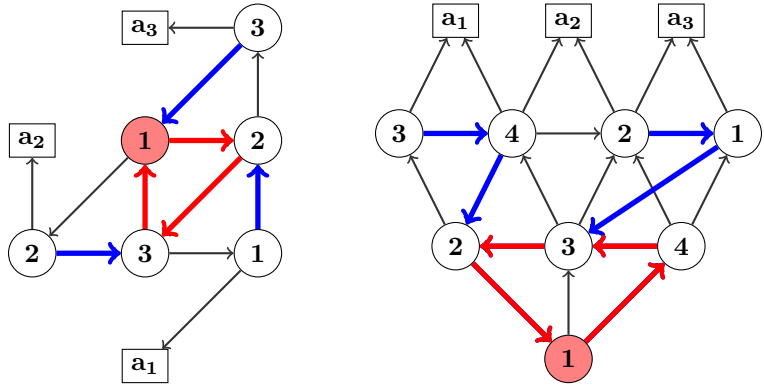

By our assumption about ($Q_C$), there exists an NE-free game $\Gamma''$ that satisfies condition ($C$).
Without loss of generality, we may assume that $\Gamma''$ is connected.
Furthermore, we may assume that  $\Gamma'$  and   $\Gamma''$  have disjoint sets of players and positions.
Let us identify all terminals of  $\Gamma'$  into
a single vertex  $v_t'$, and then identify  $v_t'$  with the initial vertex $v_0''$ of  $\Gamma''$; see Figure~\ref{fig-gamma}.

\begin{figure}[h!]
\centering
\includegraphics[width=0.5\textwidth]{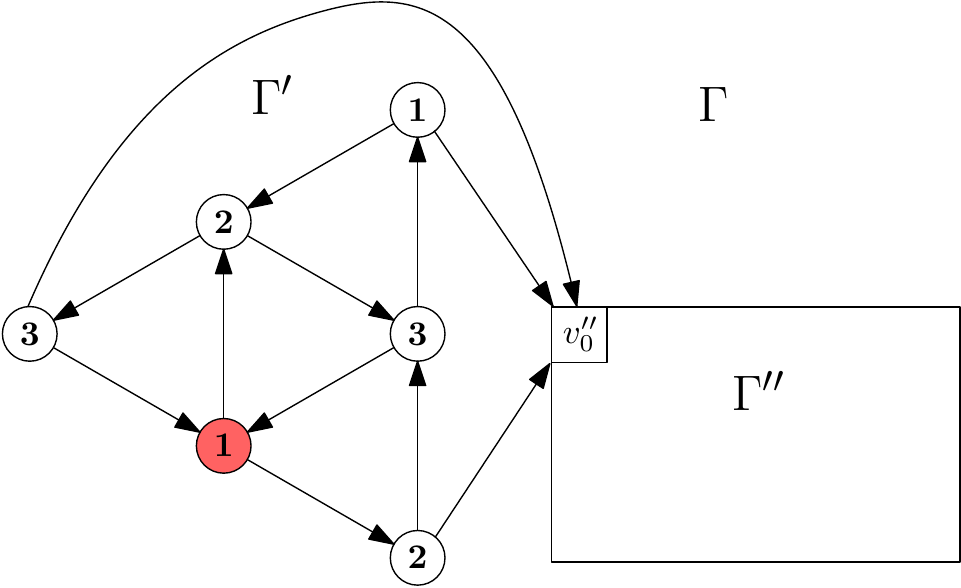}
\caption{A schematic depiction of the construction of $\Gamma$.}\label{fig-gamma}
\end{figure}

Clearly, the obtained game  $\Gamma$, having as its initial vertex the initial vertex $v_0'$ of $\Gamma'$, is connected, since
both games $\Gamma'$  and  $\Gamma''$ are connected.
Furthermore, $\Gamma$ and $\Gamma''$  have the same set of terminals.

We keep the same preferences for the players  $I''$  of  $\Gamma''$  as before, while
for the players  $I'$ of $\Gamma'$, let us consider arbitrary preferences over the outcomes of $\Gamma''$
satisfying ($C$).

First, let us show that there is a cyclic NE in $\Gamma$.
Recall that $s'$  is a cyclic NE in $\Gamma'$;
extend  $s'$ by arbitrary strategies of the players from $I''$.
The obtained strategy profile  $s$  in the game  $\Gamma$  is also a cyclic NE.
Indeed, by assumption, no player  $i' \in I'$  can reach  $v_t'$  by replacing his/her strategy
provided all players of $I' \setminus \{i'\}$  keep their strategies in $\Gamma'$.

It remains to show that no terminal NE $\sigma$ can exist in $\Gamma$.
Indeed, if the play  $P(\sigma)$  reaches a terminal of $\Gamma$ then it must pass  $v_t'$.
Since $\sigma$ is a terminal NE in $\Gamma$, its restriction to  $\Gamma''$ must be a terminal NE in  $\Gamma''$.
However, by assumption, $\Gamma''$ is NE-free. This contradiction proves our statement.
\end{proof}

\subsection{Play-once games}

In \cite{BG03}, a game was called {\em play-once} if every player controls a unique position. In such a game, the sets of mixed and independently mixed strategies of each player coincide. The game from our main example is ``almost play-once'', in the sense that only one player controls two positions, while each of the other two controls only one.

Generally speaking, the Markovian evaluation is much more common that the a priori one in the literature, with many fruitful applications, see, e.g.,~\cite{MO70}. However, for play-once DGMP games, the a priori evaluation also has its advantage.
The existence of an NE in independently mixed strategies for the a priori evaluation is guaranteed by the classical result of Nash~\cite{Nas50, Nas51},
while for the Markovian evaluation the question remains open.
Nash's arguments do not apply in this case, since the expected payoff function, as a function of probabilities, may be discontinuous.
For example, in the game represented in Figures~\ref{fig0}~and~\ref{fig-probabilities} the probability of the cyclic outcome is  $0$
unless $\gamma \delta = 1$, in which case the probability may jump to a strictly positive value.
In contrast, for the a priori evaluation, the cycle may appear with a positive probability if $\gamma$ and $\delta$ are both positive.
In this case the expected payoff is a continuous function of probabilities, hence Nash's Theorem applies.

Let us note that a subgame perfect NE in the mixed strategies may fail to exist for a play-once game,
for both Markovian and a priori evaluations, already for a $3$-person game satisfying $(C)$~\cite{BG03, BGY13}.

\vbox{The existence of an NE in pure strategies in a play-once DGMP game also remains open.
However, the answer is positive if we additionally assume condition $(C)$~\cite{BG03}.
(Interestingly, for general DGMP games the situation is opposite.)}

\subsection{Some other open questions}

It is open whether the following very general claims hold:
\begin{itemize}
\item If a DGMP game $(G,D,u,v_0)$  has no NE in pure stationary strategies then
there is no  NE in the independently mixed strategies for a game  $(G,D,w,v_0)$
with a payoff function  $w$ such that
$u_i(a)>u_i(b)$ if and only if $w_i(a)>w_i(b)$  for all  $i\in I$ and $a,b\in A$.
The problem is open for both (a) a priori and (b) Markovian evaluations.
\end{itemize}
If this claim holds for (a) then, by Nash's Theorem~\cite{Nas50, Nas51},
any play-once game has an NE in pure stationary strategies.
The latter is known to be true~\cite{BG03} in case when condition ($C$) holds.
However, (a) cannot be replaced by (b) in the above arguments.

\medskip
For the case of two terminals, $p=2$, a positive answer to question ($Q_C$) was given in~\cite{BG03}.
An attempt to extend this result to  $p=3$ was made in \cite{BR09}, but
later the authors found a flaw in the proof and, thus, already in this case ($Q_C$) remains open.

In the case of two players question ($Q$) is answered in the positive; see Section~\ref{sec:2-person-case}.
However, ($Q$) remains open for the case of two terminals.

\medskip

It is also open whether our main example is of smallest possible size
(number of positions) and if this is the case, whether it is the unique smallest one.
In particular, it is open whether the constraints on the payoff function
given in Figure~\ref{fig0} describe all NE-free games over the given digraph.

\appendix

\section*{Appendix~A: Nash equilibrium in independently mixed strategies}\label{appendix-1}

In the ``independently mixed'' setting, we denote by $S_i$ the set of independently mixed strategies of player $i \in I$;
a strategy profile  $s \in S = S_1 \times  \ldots \times S_n$  is then a collection of independently mixed strategies,
one for each player. Furthermore, the payoff values $u_i(a)$ are given, specifying the profit of player  $i$
in case outcome $a$ is realized, where  $a$  can be a terminal or the cyclic outcome.
The {\em expected payoff} of player $i \in I$  is the function  $\phi^i:S\to \RR$ defined as
\begin{equation}\label{payoff}
\phi^i(s) = \sum_{a\in A}\PP(s,a)u_i(a)\,,
\end{equation}
where $\PP(s,a)$ is the probability that outcome $a\in A$ is realized in the game,
assuming that the players play according to the independently mixed strategies specified by  $s \in S$.
A strategy profile  $s \in S$ is then called a {\em Nash equilibrium} (NE) if for each player  $i \in I$  and
for each strategy profile  $\bar{s}$  that may differ from  $s$  only in one coordinate $i$, we have
$\phi^i(s) \ge \phi^i(\bar{s})$.

\medskip
In the proofs of Theorems~\ref{thm:a-priori} and~\ref{thm:Markovian}
given in the last two subsections of the appendix, we adopt notation from Figure~\ref{fig-probabilities}, where each arc
$(v,v')$ is equipped with a number in the interval $[0,1]$, representing the probability that
the player controlling position $v$ will choose the move $(v,v')$ in $v$.
For example, for every $\alpha\in [0,1]$ the pair $(\alpha,1-\alpha)$ specifies an
(independently) mixed strategy of player $1$.
Similarly, for every $\beta \in [0,1]$ and $\gamma\in [0,1]$, the $4$-tuple $(\beta,1-\beta,\gamma, 1-\gamma)$ specifies
an independently mixed strategy of player $2$, and for every $\delta\in [0,1]$, the pair $(\delta,1-\delta)$
specifies an (independently) mixed strategy of player $3$.
Thus, every strategy profile $s\in S$ corresponds to a unique $4$-tuple $(\alpha,\beta,\gamma,\delta)\in [0,1]^4$.
We will assume this notation in the next two subsections.

\begin{figure}[h!]
\centering
\includegraphics[width=0.3\textwidth]{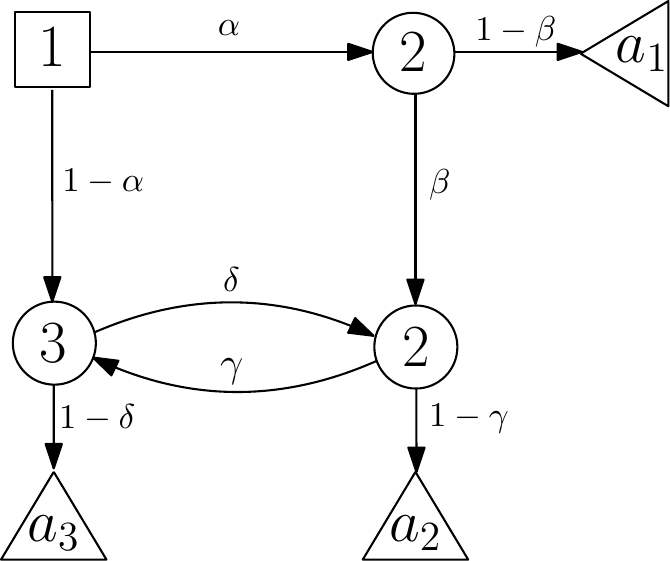}
\caption{Parameters $\alpha,\beta,\gamma,\delta$ specifying the independently mixed strategies of the players.}
\label{fig-probabilities}
\end{figure}

\section*{Appendix~B: Proof of Theorem~\ref{thm:Markovian}}
\label{sec:proof-Markovian}

Let $(G,D,u,v_0)$ be the three-person DGMP game described in Figure~\ref{fig0}.
The expected payoffs of the game are given by formula~\eqref{payoff}, which involves the outcome probabilities $\PP(s,a)$ for all
$s = (\alpha,\beta,\gamma,\delta) \in [0,1]^4$ and all $a\in A = \{a_1,a_2,a_3,c\}$, see Figure~\ref{fig-probabilities}.

We claim that the outcome probabilities are the following:
\begin{eqnarray*}
\PP(\alpha,\beta,\gamma,\delta,a_1)&=& \alpha(1-\beta)\,,\\
\PP(\alpha,\beta,\gamma,\delta,a_2)&=& \left\{
                                         \begin{array}{ll}
                                           (\alpha\beta+(1-\alpha)\delta)\frac{1-\gamma}{1-\gamma\delta}, & \hbox{if $\gamma\delta<1$;} \\
                                           0, & \hbox{otherwise\,,}
                                         \end{array}
                                       \right.\\
\PP(\alpha,\beta,\gamma,\delta,a_3)&=& \left\{
                                         \begin{array}{ll}
                                           (1-\alpha+\alpha\beta\gamma)\frac{1-\delta}{1-\gamma\delta}, & \hbox{if $\gamma\delta<1$;} \\
                                           0, & \hbox{otherwise\,,}
                                         \end{array}
                                       \right.\\
\PP(\alpha,\beta,\gamma,\delta,c)&=& \left\{
                                       \begin{array}{ll}
                                         1-\alpha+\alpha\beta, & \hbox{if $\gamma\delta = 1$;} \\
                                         0, & \hbox{otherwise.}
                                       \end{array}
                                     \right.\\
\end{eqnarray*}
The independently mixed strategies of the players defined by the $4$-tuples $(\alpha, \beta,\gamma,\delta)$ can also be viewed as the set of transition probabilites in digraph $G$, creating a Markov chain (see Figure~\ref{fig-probabilities}). The expected payoffs are computed according to the limit distribution of this Markov chain. If $\gamma\delta<1$, then the Markov chain has three absorbing classes corresponding to the three terminals and the limit probability of the $2$-cycle $c$ is zero. If $\gamma\delta=1$, then the Markov chain has two absorbing classes, terminal $a_1$ and the $2$-cycle $c$.

Digraph $G$ has a unique walk from the initial position to $a_1$, which is realized with probability $\alpha(1-\beta)$, independently of whether $\gamma\delta<1$ or not; this establishes $\PP(\alpha,\beta,\gamma,\delta,a_1)= \alpha(1-\beta)$.
Suppose that $\gamma\delta<1$. Let $v$ and $w$ be the two positions controlled by player $2$ such that $(v,w)\in E(G)$ and let $z\in V_3$ be the unique vertex controlled by player $3$. The walks from $v_0$ to $a_2$ can be classified into two types depending on whether they pass through $v$ or not.
If the walk passes through $v$, then its sequences of edges is of the form $(v_0v,vw,(wz,zw)^k, wa_2)$ for some $k\ge 0$, and the probability of its occurrence is $\alpha\beta(\gamma\delta)^k{(1-\gamma)}$. Otherwise, it is of the form $(v_0z,zw,(wz,zw)^k, wa_2)$ for some $k\ge 0$, and the probability of its occurrence is $(1-\alpha)\delta(\gamma\delta)^k{(1-\gamma)}$. Summing up these probabilities, we obtain
$$\PP(\alpha,\beta,\gamma,\delta,a_2) = (\alpha\beta+(1-\alpha)\delta)(1-\gamma)\cdot\sum_{k = 0}^\infty(\gamma\delta)^k =
(\alpha\beta+(1-\alpha)\delta)\frac{1-\gamma}{1-\gamma\delta}\,,$$
as claimed. A similar computation yields the claimed formula for the value of $\PP(\alpha,\beta,\gamma,\delta,a_3)$.

If $\gamma\delta=1$, then the probability of the occurrence of $c$ equals the probability of reaching a vertex of $c$ from the initial position, which is clearly $(1-\alpha)+(\alpha\beta)$. If the Markov chain reaches the cycle, then it will stay on it with probability one, hence in this case the limit probability of $a_2$ or $a_3$ is zero.

Suppose for a contradiction that the game has an NE, say
$s^* = (\alpha^*,\beta^*,\gamma^*,\delta^*) \in [0,1]^4$,
in independently mixed strategies for the Markovian evaluation.
We split the proof into six exhaustive cases depending on the values of $\alpha^*$,
$\beta^*$, $\gamma^*$, $\delta^*$, and show that each of them leads to a contradiction.

\medskip
\noindent{\bf Case 1: $\gamma^*\delta^* = 1$.}

\smallskip
In this case, there are only two possible outcomes: $a_1$ with probability $\alpha^*(1-\beta^*)$ and $c$ with
probability $1-\alpha^*+\alpha^*\beta^*$.

Suppose that $\alpha^*>0$ and $\beta^*<1$. Then
player $1$ would be able to strictly improve his/her expected payoff from
$\phi^1(\alpha^*,\beta^*,1,1) = \alpha^*(1-\beta^*)u_1(a_1)+ (1-\alpha^*+\alpha^*\beta^*)u_1(c)$
to
$\phi^1(0,\beta^*,1,1) = u_1(c)$, by changing his/her strategy to $\alpha = 0$.

\begin{center}
\includegraphics[width=0.7\textwidth]{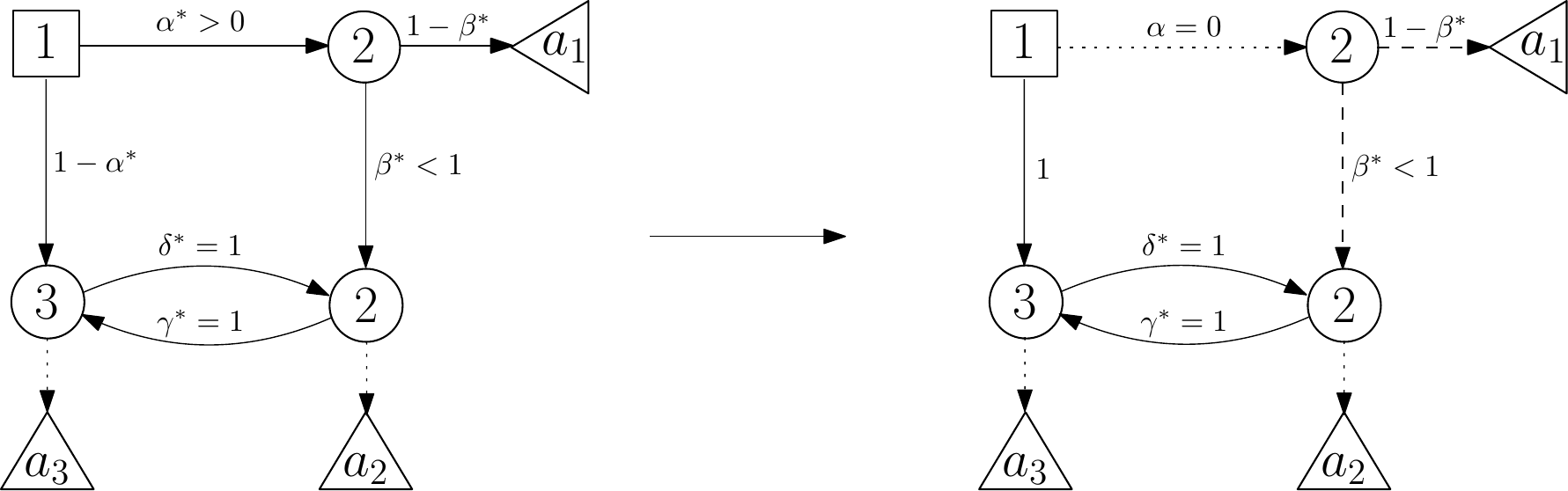}
\end{center}

This contradicts the fact that $s^*$ is an NE and shows that either $\alpha^* = 0$ or $\beta^* = 1$. In both cases,
player $2$ would be able to strictly improve his/her expected payoff from $\phi^2(\alpha^*,\beta^*,1,1) = u_2(c)$
to $\phi^2(\alpha^*,1,0,1) = u_2(a_2)$, by changing his/her strategy to $(\beta,\gamma) = (1,0)$.

\begin{center}
\includegraphics[width=0.7\textwidth]{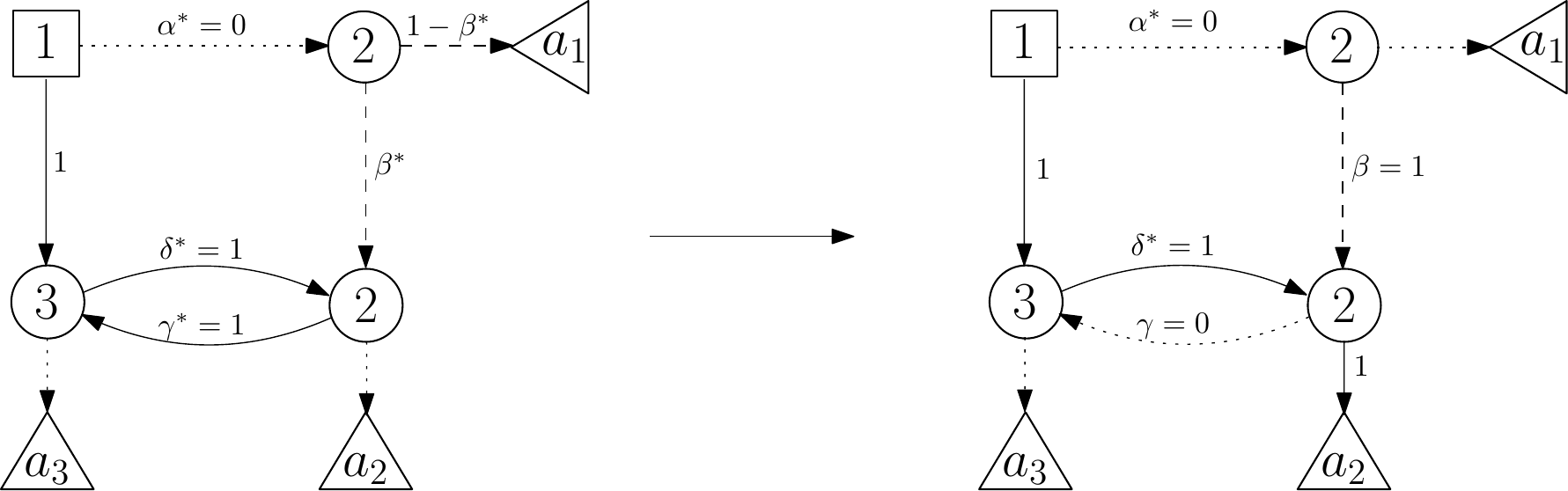}
\end{center}

\medskip
\begin{center}
\includegraphics[width=0.7\textwidth]{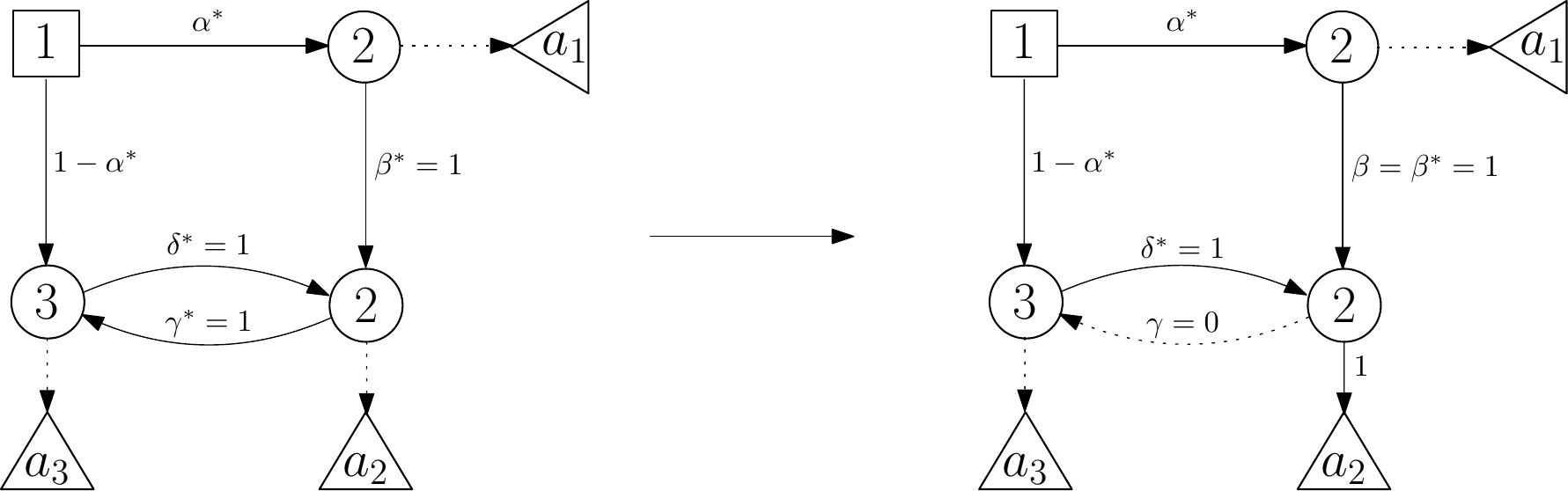}
\end{center}

This contradicts the fact that $s^*$ is an NE and settles Case 1.

\medskip
\noindent{\bf Case 2: $\alpha^*(1-\beta^*\gamma^*)(1-\delta^*) + \delta^*(1-\gamma^*)> 0$ and $\delta^*<1$.}

\smallskip
Let us first explain the reason for the conditions of this case. The
expected payoff of player $2$ is a convex combination of $u_2(a_1)$, $u_2(a_2)$, and $u_2(a_3)$, that is,
$$\phi^2(\alpha^*,\beta^*,\gamma^*,\delta^*) = \lambda u_2(a_1) +\mu u_2(a_2) +
 (1-\lambda-\mu)u_2(a_3)\,,$$
where $0\le \lambda\le 1$ and $0\le \mu\le 1$ with $\lambda+\mu\le 1$. The coefficients are given by
\begin{equation}\label{eq-convex}
\begin{aligned}
\lambda &= \PP(\alpha^*,\beta^*,\gamma^*,\delta^*,a_1) = \alpha^*(1-\beta^*)\,,\\
\mu &= \PP(\alpha^*,\beta^*,\gamma^*,\delta^*,a_2) = (\alpha^*\beta^*+(1-\alpha^*)\delta^*)\frac{1-\gamma^*}{1-\gamma^*\delta^*}\,.
\end{aligned}
\end{equation}

Recall that $u_2(a_3)>u_2(a_1)>u_2(a_2)$ and note that coefficients $\lambda$ and $\mu$ are both zero if $\beta^* = \gamma^* = 1$.
Thus, if $\lambda+\mu>0$, then player $2$ can strictly improve his/her payoff to $u_2(a_3)$,
by changing his/her strategy to $(\beta,\gamma) = (1,1)$ --- provided that also $\delta^*<1$, since if $\delta^*=1$, then
$\PP(\alpha^*,\beta^*,\gamma^*,\delta^*,a_3) = (1-\alpha^*+\alpha^*\beta^*\gamma^*)(1-\delta^*)=0$.

\begin{center}
\includegraphics[width=0.7\textwidth]{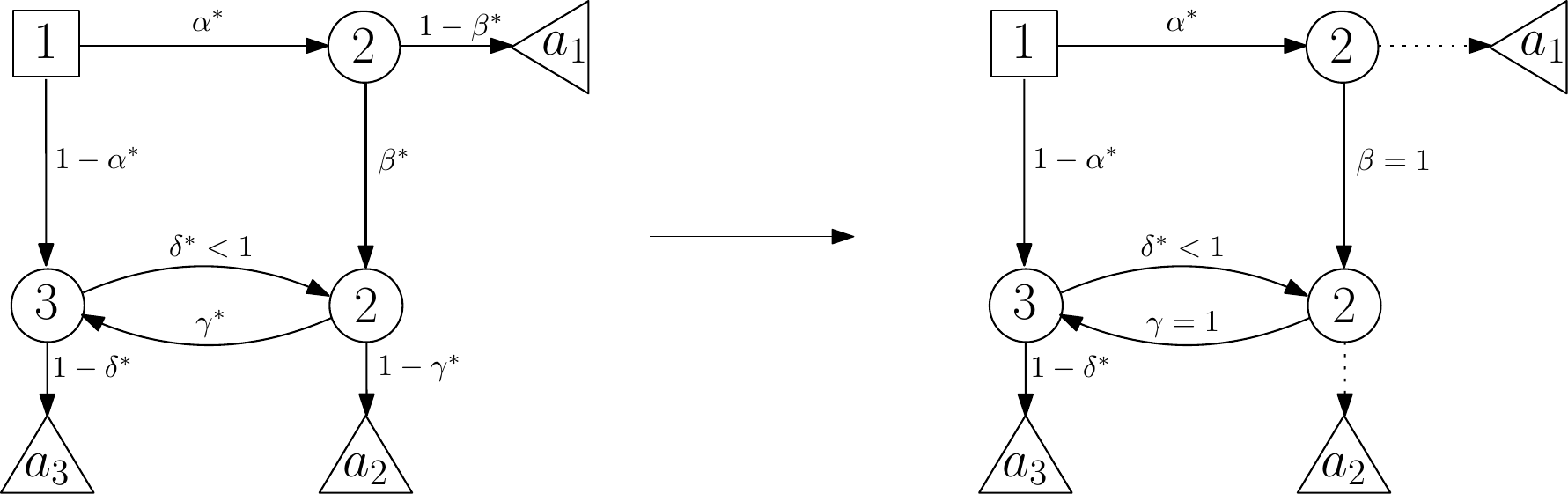}
\end{center}

Furthermore, note that $\lambda+\mu$ is positive if and only if the scaled quantity
$(\lambda+\mu)(1-\gamma^*\delta^*)$ is positive, and that
\begin{eqnarray*}
(\lambda+\mu)(1-\gamma^*\delta^*) &=&
\alpha^*(1-\beta^*)(1-\gamma^*\delta^*) + (\alpha^*\beta^*+(1-\alpha^*)\delta^*)(1-\gamma^*)\\
&=& \alpha^*(1-\beta^*\gamma^*)(1-\delta^*) + \delta^*(1-\gamma^*)\,,
\end{eqnarray*}
where the second equality can be routinely verified by separating the terms
into those involving $\alpha^*$ and the rest.

Thus, under the assumptions of Case 2, player $2$ can strictly improve his/her expected payoff from
$\phi^2(\alpha^*,\beta^*,\gamma^*,\delta^*) = \lambda u_2(a_1) +\mu u_2(a_2) + (1-\lambda-\mu)u_2(a_3)$
to $\phi^2(\alpha^*,1,1,\delta^*) = u_2(a_3)$, by changing his/her strategy to $(\beta,\gamma) = (1,1)$.
This contradicts the fact that $s^*$ is an NE and settles Case 2.

\medskip
Assuming that the conditions of Cases 1 and 2 are not satisfied, that is, $\gamma^*\delta^*<1$ and
either $\delta^*=1$ or $\alpha^*(1-\beta^*\gamma^*)(1-\delta^*) + \delta^*(1-\gamma^*)= 0$, we split the rest of the proof into four exhaustive cases.

\medskip
\noindent{\bf Case 3: $\beta^*\gamma^* = 1$ and $\delta^*<1$.}
\smallskip

In this case, player $3$ can strictly improve his/her expected payoff from
$\phi^3(\alpha^*,1,1,\delta^*) = u_3(a_3)$
to
$\phi^3(\alpha^*,1,1,1) = u_3(c)$, by changing his/her strategy to $\delta =1$.

\begin{center}
\includegraphics[width=0.7\textwidth]{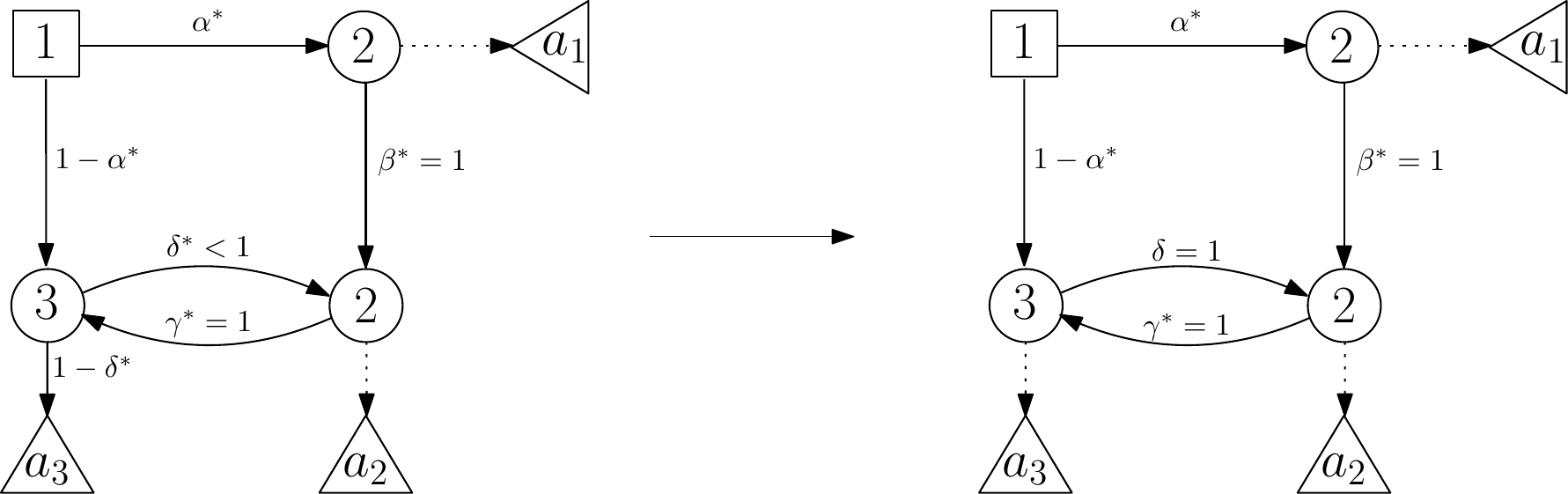}
\end{center}

This contradicts the fact that $s^*$ is an NE and settles Case 3.

\medskip
\noindent{\bf Case 4: $\alpha^* = 0$ and $\delta^* = 0$.}
\smallskip

In this case, we have $\beta^*\gamma^* \neq 1$, since otherwise
$s^*= (0,1,1,0)$ would be an NE of $(G,D,u,v_0)$ in pure strategies,
contradicting Proposition~\ref{prop:no-NE-in-pure-strategies}.
Moreover, as the outcome of the game is $a_3$ with probability one, the expected payoff of player $1$ is $u_1(a_3)$, the worst possible for him/her.
If player $1$ changes his/her strategy to $\alpha =1$, then his/her expected outcome would become
$\phi^1(1,\beta^*,\gamma^*,0) = \lambda u_1(a_1) +\mu u_1(a_2) +
 (1-\lambda-\mu)u_1(a_3)$, with $\lambda$ and $\mu$ as given by~\eqref{eq-convex} with $\alpha^* = 1$.
For $\alpha^* = 1$, $\beta^*\gamma^* < 1$ (which is true as argued above) and $\delta^* = 0$, we have
$\alpha^*(1-\beta^*\gamma^*)(1-\delta^*) + \delta^*(1-\gamma^*)>0$ and hence $\lambda+\mu>0$ (recall Case 2).
Thus, $\phi^1(1,\beta^*,\gamma^*,0)>\phi^1(0,\beta^*,\gamma^*,0) = u_1(a_3)$, hence
player $1$ can strictly improve his/her expected payoff by changing his/her strategy to $\alpha =1$.

\begin{center}
\includegraphics[width=0.7\textwidth]{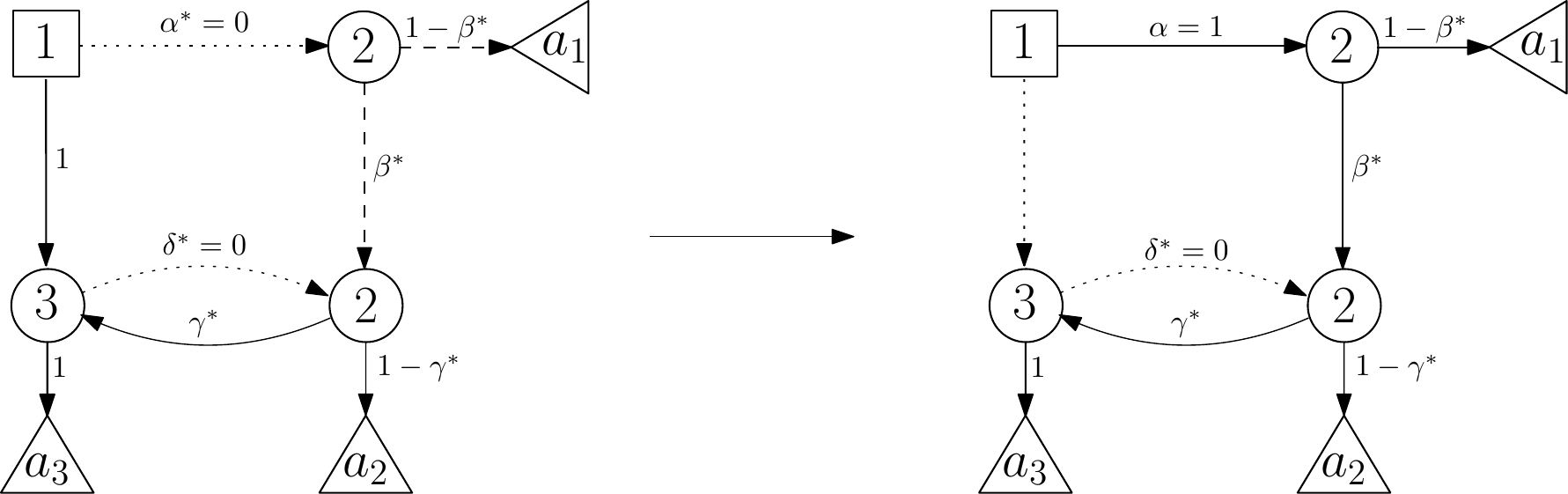}
\end{center}

This contradicts the fact that $s^*$ is an NE and settles Case 4.

\medskip
\noindent{\bf Case 5: $\gamma^*<1$ and $\delta^* = 1$.}
\smallskip

We split this case into three further subcases depending on the values of $\alpha^*$ and $\beta^*$ (assuming, in all the subcases, that
$\gamma^*<1$ and $\delta^* = 1$).

\medskip
\noindent{\bf Subcase 5.1: $\alpha^*>0$ and $\beta^* = 1$.}
\smallskip

In this case, the outcome of the game is $a_2$ with probability one, hence the expected payoff of player $2$ is $u_2(a_2)$.
If player $2$ changes his/her strategy to $(\beta,\gamma) = (0,\gamma^*)$, then his/her expected outcome would become
$\phi^2(\alpha^*,0,\gamma^*,1) = \alpha^*u_2(a_1)+(1-\alpha^*)u_2(a_2)$, which, since $\alpha^*>0$, is a strict improvement over
$\phi^2(\alpha^*,1,\gamma^*,1) = u_2(a_2)$.

\begin{center}
\includegraphics[width=0.7\textwidth]{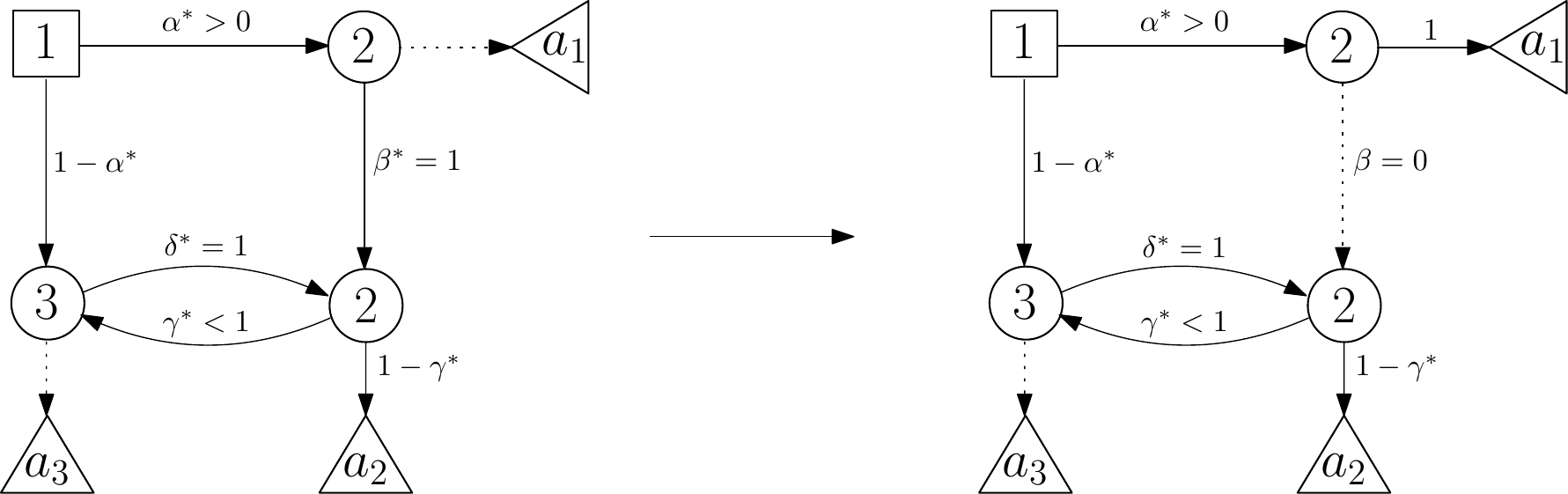}
\end{center}

This contradicts the fact that $s^*$ is an NE and settles Subcase  5.1.

\medskip
\noindent{\bf Subcase 5.2: $\alpha^*>0$ and $\beta^* < 1$.}
\smallskip

In this case, the expected payoff of player $1$ is $\phi^1(\alpha^*,\beta^*,\gamma^*,1)=\alpha^*(1-\beta^*)u_1(a_1)+(1-\alpha^*+\alpha^*\beta^*)u_1(a_2)$, a convex combination of $u_1(a_1)$ and $u_1(a_2)$, with strictly positive coefficient at $u_1(a_1)$.
If player $1$ changes his/her strategy to $\alpha = 0$, then his/her expected outcome would become
$\phi^1(0,\beta^*, \gamma^*,1) = u_1(a_2)$, a strict improvement over $\phi^1(\alpha^*,\beta^*,\gamma^*,1)$.

\begin{center}
\includegraphics[width=0.7\textwidth]{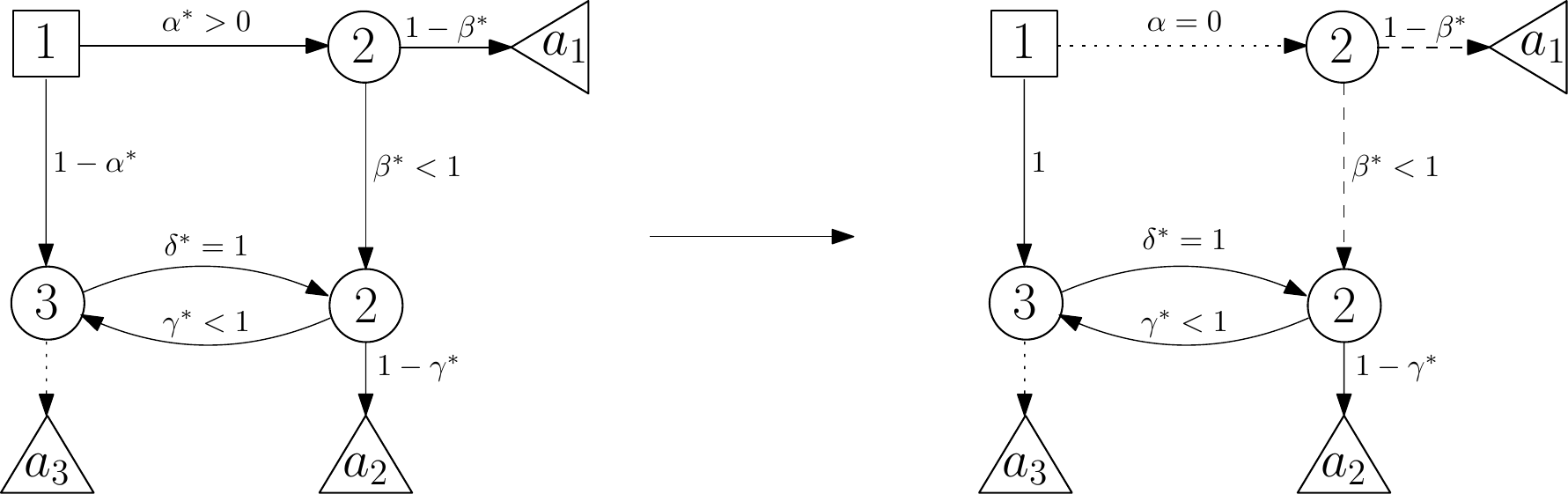}
\end{center}

This contradicts the fact that $s^*$ is an NE and settles Subcase  5.2.

\medskip
\noindent{\bf Subcase 5.3: $\alpha^*=0$.}
\smallskip

In this case, player $3$ can strictly improve his/her expected payoff from
$\phi^3(0,\beta^*,\gamma^*,1)=u_3(a_2)$
to
$\phi^3(0,\beta^*,\gamma^*,0) = u_3(a_3)$, by changing his/her strategy to $\delta = 0$.

\begin{center}
\includegraphics[width=0.7\textwidth]{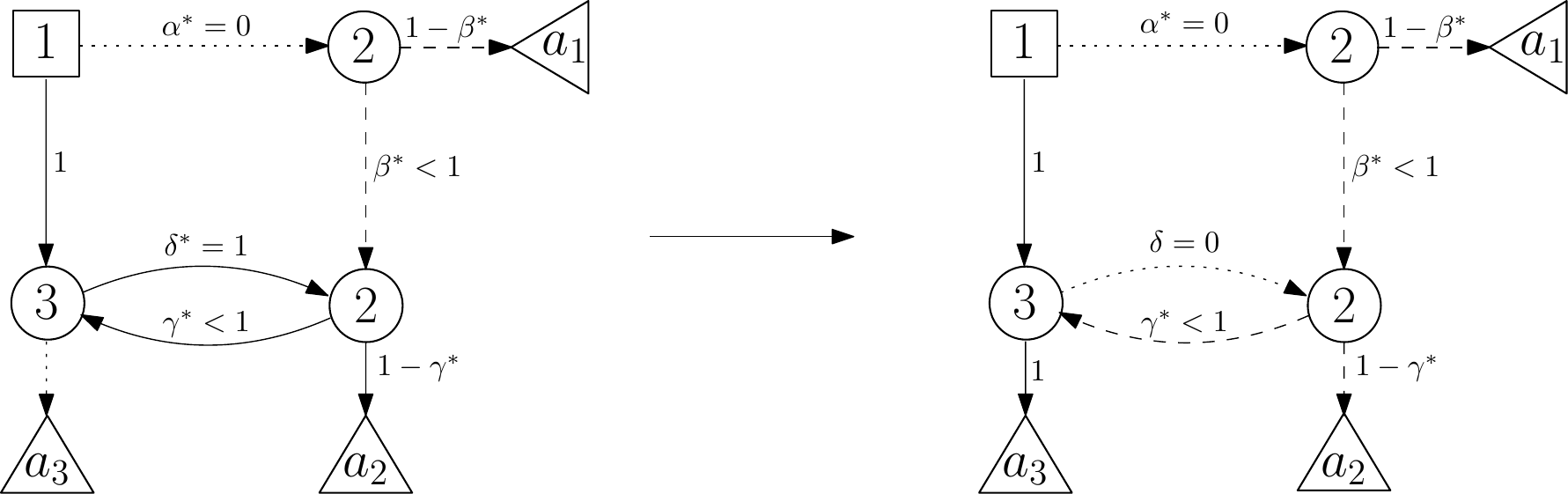}
\end{center}

This contradicts the fact that $s^*$ is an NE and settles Subcase 5.3 and with it Case 5.

\medskip
\noindent{\bf Case 6: $\alpha^* = 0$, $\gamma^* = 1$, and $\delta^*<1$.}
\smallskip

In this case, player $3$ can strictly improve his/her expected payoff from
$\phi^3(0,\beta^*,1,\delta^*)=u_3(a_3)$
to
$\phi^3(0,\beta^*,1,1) = u_3(c)$, by changing his/her strategy to $\delta = 1$.

\begin{center}
\includegraphics[width=0.7\textwidth]{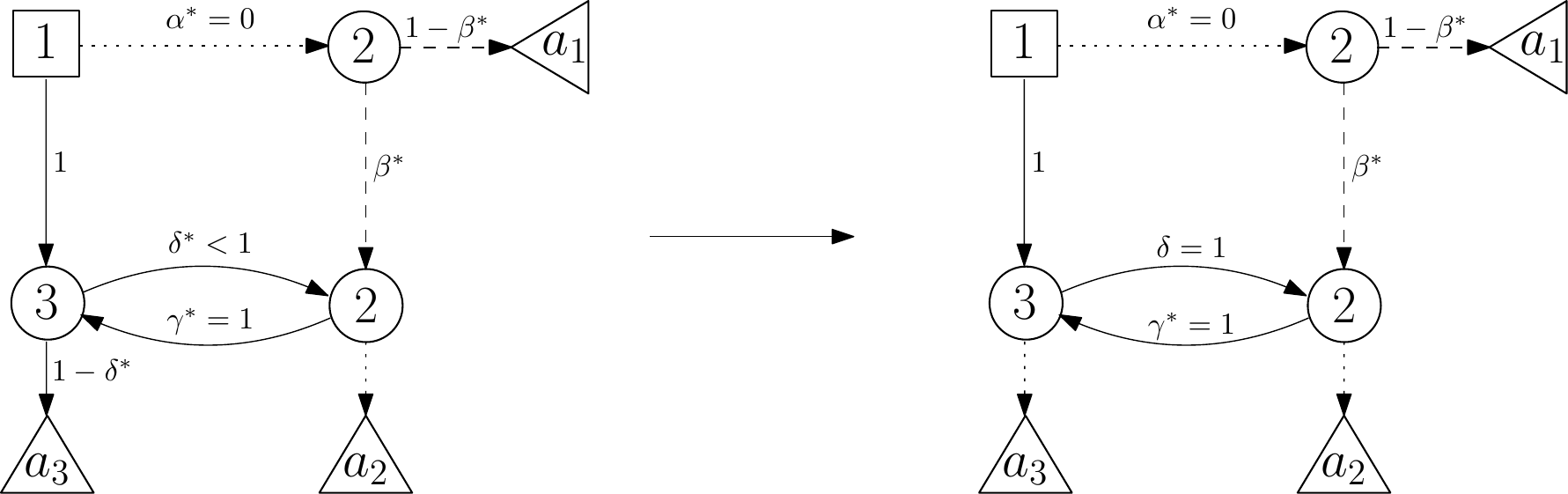}
\end{center}

This contradicts the fact that $s^*$ is an NE and settles Case 6.

\medskip
It remains to show that Cases 1--6 exhaust all possibilities.
Perhaps the easiest way to observe this is to notice that if none of the assumptions of Cases 1, 3, 4, 5, and 6 hold,
then Case 2 must hold. Indeed, assuming that Cases 1 and 5 do not hold, we infer that $\delta^*<1$.
If in addition that Case 3 does not hold, then we have $\beta^*\gamma^* < 1$. If Case 4 fails, then we have $\alpha^*>0$ or $\delta^*>0$,
and if Case 6 fails, then $\alpha^*>0$ or $\gamma^*<1$.
Together, all these inequalities imply that $\alpha^*(1-\beta^*\gamma^*)(1-\delta^*)>0$ or
$\delta^*(1-\gamma^*)>0$, which together with $\delta^*<1$ shows that the assumptions of Case 2 hold.
This completes the proof.

\section*{Appendix~C: Proof of Theorem~\ref{thm:a-priori}}\label{sec:proof-a-priori}

Let $(G,D,u,v_0)$ be the three-person DGMP game described in Figure~\ref{fig0} and suppose in addition that the payoff function $u$ satisfies conditions (1)--(\ref{assumption3}). Recall that the expected payoffs of the game are given by formula~\eqref{payoff}, which involves the outcome probabilities $\PP(s,a)$ for all $s = (\alpha,\beta,\gamma,\delta) \in [0,1]^4$ and all $a\in A = \{a_1,a_2,a_3,c\}$, see Figure~\ref{fig-probabilities}.

Analyzing the directed paths from the initial position to each terminal position and to the cycle in $G$, we obtain:
\begin{eqnarray*}
\PP(\alpha,\beta,\gamma,\delta,a_1)&=& \alpha(1-\beta)\,,\\
\PP(\alpha,\beta,\gamma,\delta,a_2)&=& (\alpha\beta+(1-\alpha)\delta)(1-\gamma)\,,\\
\PP(\alpha,\beta,\gamma,\delta,a_3)&=& (1-\alpha+\alpha\beta\gamma)(1-\delta)\,,\\
\PP(\alpha,\beta,\gamma,\delta,c)&=& (1-\alpha+\alpha\beta)\gamma\delta\,.
\end{eqnarray*}
Suppose for a contradiction that the game has an NE,
say $s^* = (\alpha^*,\beta^*,\gamma^*,\delta^*) \in [0,1]^4$, in independently mixed strategies
for the a priori evaluation.
Note that
\begin{equation*}
0<\alpha^*<1 \quad \Longrightarrow \quad \frac{\partial \phi^1}{\partial \alpha}(\alpha^*,\beta^*,\gamma^*,\delta^*) = 0\,.
\end{equation*}
Indeed, if $0<\alpha^*<1$, then the definition of an NE implies that
$\alpha^*$ is a local maximum of $f^1$, where $f^1:[0,1]\to \RR$ is
the restriction of $\phi^1$, the expected payoff function of player $1$,
to $\beta = \beta^*$, $\gamma = \gamma^*$, $\delta = \delta^*$.
(That is, for all $\alpha\in [0,1]$, we have $f^1(\alpha) = \phi^1(\alpha,\beta^*,\gamma^*,\delta^*)$.)
Since $f^1$ is a continuously differentiable function of $\alpha\in [0,1]$, we infer that $$\frac{df^1}{d\alpha}(\alpha^*)= 0$$
or, equivalently,
$$\frac{\partial \phi^1}{\partial \alpha}(\alpha^*,\beta^*,\gamma^*,\delta^*) = 0\,,$$
as claimed.

\begin{sloppypar}
Similarly, we infer that
\begin{equation*}
\alpha^* = 0 \quad \Longrightarrow \quad \frac{\partial \phi^1}{\partial \alpha}(\alpha^*,\beta^*,\gamma^*,\delta^*) \le 0\,,
\end{equation*}
where $\frac{\partial \phi^1}{\partial \alpha}(\alpha^*,\beta^*,\gamma^*,\delta^*)$ denotes the right partial derivative of
$\phi^1$ with respect to $\alpha$ evaluated at $(\alpha^*,\beta^*,\gamma^*,\delta^*) = (0,\beta^*,\gamma^*,\delta^*)$, and
\begin{equation*}
\alpha^* = 1 \quad \Longrightarrow \quad \frac{\partial \phi^1}{\partial \alpha}(\alpha^*,\beta^*,\gamma^*,\delta^*) \ge 0\,,
\end{equation*}
where $\frac{\partial \phi^1}{\partial \alpha}(\alpha^*,\beta^*,\gamma^*,\delta^*)$
denotes the left partial derivative of $\phi^1$ with respect to $\alpha$
evaluated at $(\alpha^*,\beta^*,\gamma^*,\delta^*) = (1,\beta^*,\gamma^*,\delta^*)$.
We can write these three implications concisely as:
\begin{equation}
\label{phi1alpha}
\left\{
  \begin{array}{c}
    \alpha^* = 0\\
    0<\alpha^* <1\\
    \alpha^* = 1
  \end{array}
\right\} \quad \Longrightarrow \quad
\frac{\partial \phi^1}{\partial \alpha}(\alpha^*,\beta^*,\gamma^*,\delta^*) \left\{
  \begin{array}{c}
    \le  0\\
    = 0\\
    \ge 0
  \end{array}
\right\} \,.
\end{equation}
In a similar way, we obtain analogous implications involving values of the other three components of $s^*$:
\begin{equation*}
%\label{phi2beta}
\left\{
  \begin{array}{c}
    \beta^* = 0\\
    0<\beta^* <1\\
    \beta^* = 1
  \end{array}
\right\} \quad \Longrightarrow \quad
\frac{\partial \phi^2}{\partial \beta}(\alpha^*,\beta^*,\gamma^*,\delta^*) \left\{
  \begin{array}{c}
    \le  0\\
    = 0\\
    \ge 0
  \end{array}
\right\} \,,
\end{equation*}
\begin{equation*}
%\label{phi2beta}
\left\{
  \begin{array}{c}
    \gamma^* = 0\\
    0<\gamma^* <1\\
    \gamma^* = 1
  \end{array}
\right\} \quad \Longrightarrow \quad
\frac{\partial \phi^2}{\partial \gamma}(\alpha^*,\beta^*,\gamma^*,\delta^*) \left\{
  \begin{array}{c}
    \le  0\\
    = 0\\
    \ge 0
  \end{array}
\right\} \,,
\end{equation*}
\begin{equation*}
%\label{phi2beta}
\left\{
  \begin{array}{c}
    \delta^* = 0\\
    0<\delta^* <1\\
    \delta^* = 1
  \end{array}
\right\} \quad \Longrightarrow \quad
\frac{\partial \phi^3}{\partial \delta}(\alpha^*,\beta^*,\gamma^*,\delta^*) \left\{
  \begin{array}{c}
    \le  0\\
    = 0\\
    \ge 0
  \end{array}
\right\} \,.
\end{equation*}
\end{sloppypar}

Since for all $i\in I$ and all $(\alpha,\beta,\gamma,\delta)\in [0,1]^4$ we have
$$\phi^i(\alpha,\beta,\gamma,\delta) = \sum_{a\in A}\PP(\alpha,\beta,\gamma,\delta,a)u_i(a)\,,$$
we obtain
\begin{equation}\label{phi1alpha-derivative}
\frac{\partial\phi^1(\alpha,\beta,\gamma,\delta)}{\partial \alpha}
= \sum_{a\in A}\frac{\partial\PP(\alpha,\beta,\gamma,\delta,a)}{\partial \alpha}u_1(a)
\end{equation}
and similar expressions hold for
$\frac{\partial\phi^2}{\partial \beta}$,
$\frac{\partial\phi^2}{\partial \gamma}$,
$\frac{\partial\phi^3}{\partial \delta}$.
Thus, we only need to compute the partial derivatives of the form
$\frac{\partial\PP(\alpha,\beta,\gamma,\delta,a)}{\partial \eta}$
for all $\eta\in \{\alpha,\beta,\gamma,\delta\}$ and all $a\in A$.

\begin{center}
\begin{minipage}[b]{0.46\textwidth}
For $\eta = \alpha$, we obtain:
\begin{eqnarray}
\label{prob-alpha-derivatives-a1}\frac{\partial\PP(\alpha,\beta,\gamma,\delta,a_1)}{\partial \alpha} &=& 1-\beta\,,\\
\label{prob-alpha-derivatives-a2}\frac{\partial\PP(\alpha,\beta,\gamma,\delta,a_2)}{\partial \alpha} &=& (\beta-\delta)(1-\gamma)\,,\\
\label{prob-alpha-derivatives-a3}\frac{\partial\PP(\alpha,\beta,\gamma,\delta,a_3)}{\partial \alpha} &=& -(1-\beta\gamma)(1-\delta)\,,\\
\label{prob-alpha-derivatives-c}\frac{\partial\PP(\alpha,\beta,\gamma,\delta,c)}{\partial \alpha} &=& -(1-\beta)\gamma\delta\,.
\end{eqnarray}
\end{minipage}
\hspace{1cm}
\begin{minipage}[b]{0.46\textwidth}
For $\eta = \beta$, we obtain:
\begin{eqnarray*}
\frac{\partial\PP(\alpha,\beta,\gamma,\delta,a_1)}{\partial \beta} &=& -\alpha\,,\\
\frac{\partial\PP(\alpha,\beta,\gamma,\delta,a_2)}{\partial \beta} &=& \alpha(1-\gamma)\,,\\
\frac{\partial\PP(\alpha,\beta,\gamma,\delta,a_3)}{\partial \beta} &=& \alpha\gamma(1-\delta)\,,\\
\frac{\partial\PP(\alpha,\beta,\gamma,\delta,c)}{\partial \beta} &=& \alpha\gamma\delta\,.
\end{eqnarray*}
\end{minipage}
\end{center}

\begin{center}
\begin{minipage}[b]{0.46\textwidth}
For $\eta = \gamma$, we obtain:
\begin{eqnarray*}
\frac{\partial\PP(\alpha,\beta,\gamma,\delta,a_1)}{\partial \gamma} &=& 0\,,\\
\frac{\partial\PP(\alpha,\beta,\gamma,\delta,a_2)}{\partial \gamma} &=& -(\alpha\beta+(1-\alpha)\delta)\,,\\
\frac{\partial\PP(\alpha,\beta,\gamma,\delta,a_3)}{\partial \gamma} &=& \alpha\beta(1-\delta)\,,\\
\frac{\partial\PP(\alpha,\beta,\gamma,\delta,c)}{\partial \gamma} &=& (1-\alpha+\alpha\beta)\delta\,.
\end{eqnarray*}
\end{minipage}
\hspace{1cm}
\begin{minipage}[b]{0.46\textwidth}
For $\eta = \delta$, we obtain:
\begin{eqnarray*}
\frac{\partial\PP(\alpha,\beta,\gamma,\delta,a_1)}{\partial \delta} &=& 0\,,\\
\frac{\partial\PP(\alpha,\beta,\gamma,\delta,a_2)}{\partial \delta} &=& (1-\alpha)(1-\gamma)\,,\\
\frac{\partial\PP(\alpha,\beta,\gamma,\delta,a_3)}{\partial \delta} &=& -(1-\alpha+\alpha\beta\gamma)\,,\\
\frac{\partial\PP(\alpha,\beta,\gamma,\delta,c)}{\partial \delta} &=& (1-\alpha+\alpha\beta)\gamma\,.
\end{eqnarray*}
\end{minipage}
\end{center}

Recall that $u_1(a_3) = u_2(c) = u_3(a_2) = 0$.
Plugging the derivatives given by~equations~\eqref{prob-alpha-derivatives-a1}--\eqref{prob-alpha-derivatives-c}
into~\eqref{phi1alpha-derivative} yields, together with
implication~\eqref{phi1alpha} and the assumption $u_1(a_3) = 0$, the following:
\begin{enumerate}
\item[($P_1$)]
\begin{equation*}
\left\{
  \begin{array}{c}
    \alpha^* = 0\\
    0<\alpha^* <1\\
    \alpha^* = 1
  \end{array}
\right\} \quad \Longrightarrow \quad
(1-\beta^*)u_1(a_1)
+ (\beta^*-\delta^*)(1-\gamma^*)u_1(a_2)
 -(1-\beta^*)\gamma^*\delta^* u_1(c)
\left\{
  \begin{array}{c}
    \le  0\\
    = 0\\
    \ge 0
  \end{array}
\right\} \,.
\end{equation*}
\end{enumerate}
Similarly, we get:
\begin{enumerate}
\item[($P_2$)]
\begin{equation*}
\left\{
  \begin{array}{c}
    \beta^* = 0\\
    0<\beta^* <1\\
    \beta^* = 1
  \end{array}
\right\} \quad \Longrightarrow \quad
-\alpha^* u_2(a_1)
+ \alpha^*(1-\gamma^*)u_2(a_2)
+\alpha^*\gamma^*(1-\delta^*)u_2(a_3)
\left\{
  \begin{array}{c}
    \le  0\\
    = 0\\
    \ge 0
  \end{array}
\right\} \,,
\end{equation*}
\begin{equation*}
\left\{
  \begin{array}{c}
    \gamma^* = 0\\
    0<\gamma^* <1\\
    \gamma^* = 1
  \end{array}
\right\} \quad \Longrightarrow \quad
-(\alpha^*\beta^*+(1-\alpha^*)\delta^*)u_2(a_2)
+\alpha^*\beta^*(1-\delta^*)u_2(a_3)
 \left\{
  \begin{array}{c}
    \le  0\\
    = 0\\
    \ge 0
  \end{array}
\right\} \,,
\end{equation*}
\item[($P_3$)]
\begin{equation*}
\left\{
  \begin{array}{c}
    \delta^* = 0\\
    0<\delta^* <1\\
    \delta^* = 1
  \end{array}
\right\} \quad \Longrightarrow \quad
-(1-\alpha^*+\alpha^*\beta^*\gamma^*)u_3(a_3)
+(1-\alpha^*+\alpha^*\beta^*)\gamma^*u_3(c)
\left\{
  \begin{array}{c}
    \le  0\\
    = 0\\
    \ge 0
  \end{array}
\right\} \,.
\end{equation*}
\end{enumerate}

\begin{observation1}\label{observation1}
If $(1-\beta^*)\gamma^*\delta^*= 0$ {and $\beta^*\ge \delta^*$}, then either $\beta^*= 1$ and $(1-\delta^*)(1-\gamma^*) = 0$, or $\alpha^* = 1$.
\end{observation1}

\noindent
{\it Proof.} Suppose that $(1-\beta^*)\gamma^*\delta^*= 0$ and $\alpha^*<1$.
By $(P_1)$, we have $(1-\beta^*)u_1(a_1) + (\beta^*-\delta^*)(1-\gamma^*)u_1(a_2)\le 0$.
Since $u_1(a_1)>0$ and $u_1(a_2)>0$, while
$1-\beta^*\ge 0$ and
$(\beta^*-\delta^*)(1-\gamma^*)\ge 0$, we infer $1-\beta^*= 0$ and $(\beta^*-\delta^*)(1-\gamma^*) = 0$.
\hfill$\blacktriangle$

\medskip
\begin{observation2a}\label{observation2a}
If $\alpha^*>0$ and $\gamma^*(1-\delta^*) = 0$, then $\beta^* = 0$.
\end{observation2a}

\noindent
{\it Proof.}
Suppose for a contradiction that $\alpha^*>0$, $\gamma^*(1-\delta^*) = 0$, and $\beta^*>0$.
By the first implication in $(P_2)$, we obtain $-\alpha^* u_2(a_1) + \alpha^*(1-\gamma^*)u_2(a_2)\ge 0$, which further simplifies to
$(1-\gamma^*)u_2(a_2)\ge u_2(a_1)$, contrary to the assumption $u_2(a_1)>u_2(a_2)>0$. \hfill$\blacktriangle$

\medskip
\begin{observation2b}\label{observation2b}
If $\alpha^*\beta^*(1-\delta^*) = 0$, then either $\alpha^*\beta^*+(1-\alpha^*)\delta^* = 0$ or $\gamma^* = 0$.
\end{observation2b}

\noindent
{\it Proof.}
Suppose that $\alpha^*\beta^*(1-\delta^*) = 0$ and $\gamma^* > 0$.
By the second implication in $(P_2)$, we obtain $-(\alpha^*\beta^*+(1-\alpha^*)\delta^*)u_2(a_2) \ge 0$.
Since $u_2(a_2) >0$, this yields $\alpha^*\beta^*+(1-\alpha^*)\delta^*\le 0$. Since both summands are non-negative,
it follows that $\alpha^*\beta^*+(1-\alpha^*)\delta^* = 0$.\hfill$\blacktriangle$

\medskip
\begin{observation3a}\label{observation3a}
If $(1-\alpha^*+\alpha^*\beta^*)\gamma^* = 0$, then either $1-\alpha^*+\alpha^*\beta^*\gamma^* = 0$ or $\delta^* = 0$.
\end{observation3a}

\noindent
{\it Proof.}
Suppose that $(1-\alpha^*+\alpha^*\beta^*)\gamma^* = 0$ and $\delta^* > 0$.
By $(P_3)$, we have $-(1-\alpha^*+\alpha^*\beta^*\gamma^*)u_3(a_3)\ge 0$.
Since $u_3(a_3) > 0$, this implies
$(1-\alpha^*)+\alpha^*\beta^*\gamma^*\le 0$. Since both summands are non-negative,
it follows that $1-\alpha^*+\alpha^*\beta^*\gamma^* = 0$.\hfill$\blacktriangle$

\medskip
\begin{observation3b}\label{observation3b}
If $\gamma^* = 1$, then either $1-\alpha^*+\alpha^*\beta^* = 0$ or $\delta^* = 1$.
\end{observation3b}

\noindent
{\it Proof.}
Suppose that $\gamma^* = 1$ and $\delta^* < 1$.
By $(P_3)$, we have
$(1-\alpha^*+\alpha^*\beta^*)u_3(c)\le (1-\alpha^*+\alpha^*\beta^*)u_3(a_3)$, or, equivalently,
$(1-\alpha^*+\alpha^*\beta^*)(u_3(c)-u_3(a_3))\le 0$.
Since $u_3(c)-u_3(a_3)>0$, this implies
$(1-\alpha^*)+\alpha^*\beta^*\le 0$, and, since both summands are non-negative,
that $1-\alpha^*+\alpha^*\beta^*= 0$.\hfill$\blacktriangle$

\medskip
We now prove several claims.

\begin{claim}\label{claim1}
If $\alpha^*>0$ and $\beta^*>0$, then $\gamma^*>1/2$.
\end{claim}

\noindent
{\it Proof}.
Since $\alpha^*>0$ and $\beta^*>0$, we infer from~($P_2$) that
\begin{equation*}
u_2(a_1) \le (1-\gamma^*)u_2(a_2) +\gamma^*(1-\delta^*)u_2(a_3)\,.
\end{equation*}
Moreover, assumption~\eqref{assumption1} is equivalent to $u_2(a_1)>(u_2(a_2)+u_2(a_3))/2$.
Combining the two inequalities and separating the terms involving $u_2(a_2)$ from those
involving $u_2(a_3)$, we obtain
$$(\gamma^*-1/2)u_2(a_2)<(\gamma^*-\gamma^*\delta^*-1/2)u_2(a_3)\,,$$
which, using the facts that $\gamma^*\delta^*\ge 0$ and $u_2(a_3)>0$, implies
$$(\gamma^*-1/2)u_2(a_2)<(\gamma^*-1/2)u_2(a_3)\,,$$ or, equivalently,
$$(\gamma^*-1/2)(u_2(a_3)-u_2(a_2))>0\,.$$
Since $u_2(a_2)<u_2(a_3)$, we infer $\gamma^*>1/2$, as claimed. \hfill$\blacktriangle$
\medskip

\begin{claim}\label{claim2}
If $\delta^*<1$ and $1-\alpha^*+\alpha^*\beta^*\gamma^*>0$, then $\gamma^*< 1/2$.
\end{claim}

\noindent
{\it Proof}.
Since $\delta^*<1$ and $1-\alpha^*+\alpha^*\beta^*\gamma^*>0$,
we infer from~($P_3$) that
\begin{equation}
\label{eq1}
u_3(a_3) \ge \frac{(1-\alpha^*+\alpha^*\beta^*)\gamma^*}{1-\alpha^*+\alpha^*\beta^*\gamma^*}\cdot u_3(c)\,.
\end{equation}
Furthermore, inequality~\eqref{assumption2} means that $u_3(a_3)<u_3(c)/2$.
Combining the two inequalities and canceling out the positive factor $u_3(c)$, we obtain
$$\frac{(1-\alpha^*+\alpha^*\beta^*)\gamma^*}{1-\alpha^*+\alpha^*\beta^*\gamma^*}<1/2\,,$$
which, using $1-\alpha^*+\alpha^*\beta^*\gamma^*>0$, simplifies to
$$2\gamma^*(1-\alpha^*+\alpha^*\beta^*)<1-\alpha^*+\alpha^*\beta^*\gamma^*\,.$$
Suppose for a contradiction that $\gamma^*\ge 1/2$.
Then $1-\alpha^*+\alpha^*\beta^*<1-\alpha^*+\alpha^*\beta^*\gamma^*$,
which implies $\alpha^*\beta^*(1-\gamma^*)<0$, a contradiction. We conclude that $\gamma^*<1/2$, as claimed. \hfill$\blacktriangle$

\medskip

\begin{claim}\label{claim3}
If $\alpha^*>0$, $\beta^*>0$, and $\delta^*<1$, then $\alpha^* = 1$ and $\gamma^* = 0$.
\end{claim}

\noindent
{\it Proof.}
Suppose that  $\alpha^*>0$, $\beta^*>0$, and $\delta^*<1$.
From Claims~\ref{claim1} and~\ref{claim2} we infer that $(1-\alpha^*)+\alpha^*\beta^*\gamma^*\le 0$.
Since both summands are non-negative, both are in fact zero, that is, $\alpha^* = 1$ and $\beta^*\gamma^* = 0$.
Since $\beta^*>0$, we obtain $\gamma^* = 0$.\hfill$\blacktriangle$

\medskip

\begin{claim}\label{claim4}
If $\beta^*>0$ and $0<\delta^*<1$, then $\alpha^* = \gamma^* = 0$.
\end{claim}

\noindent
{\it Proof.}
Suppose that $\beta^*>0$ and $0<\delta^*<1$.
If $\alpha^*>0$, then Claim~\ref{claim3} implies that $\gamma^* = 0$, from which we obtain, using Observation 2A,
a contradictory $\beta^* = 0$. Therefore, $\alpha^*=0$.
Since $\alpha^*\beta^*+(1-\alpha^*)\delta^* = \delta^*>0$, we can now use Observation 2B to infer
that $\gamma^* = 0$.\hfill$\blacktriangle$

\medskip
\begin{claim}\label{claim5}
Either $\beta^*= 0$ or $\delta^*\in \{0,1\}$.
\end{claim}

\noindent
\noindent
{\it Proof.}
Suppose for a contradiction that $\beta^*>0$ and $0<\delta^*<1$.
By Claim~\ref{claim4}, we obtain $\alpha^* = \gamma^* = 0$.
Since $\delta^*>0$, this contradicts Observation 3A.
\hfill$\blacktriangle$

\medskip
\begin{claim}\label{claim6}
$\beta^* > 0$.
\end{claim}

\noindent
{\it Proof.}
Suppose for a contradiction that $\beta^* = 0$. We infer from Observation 2B that
$(1-\alpha^*)\delta^* = 0$ or $\gamma^* = 0$.

Suppose that $\delta^* = 0$. Then $\alpha^* = 1$ by Observation 1.
But now, player $2$ can strictly improve his/her expected payoff from $\phi^2(1,0,\gamma^*,0) = u_2(a_1)$
to $\phi^2(1,1,1,0) = u_2(a_3)$ by changing his/her strategy to $(\beta,\gamma) = (1,1)$,
contrary to the fact that $s^*$ is an NE.

\begin{center}
\includegraphics[width=0.7\textwidth]{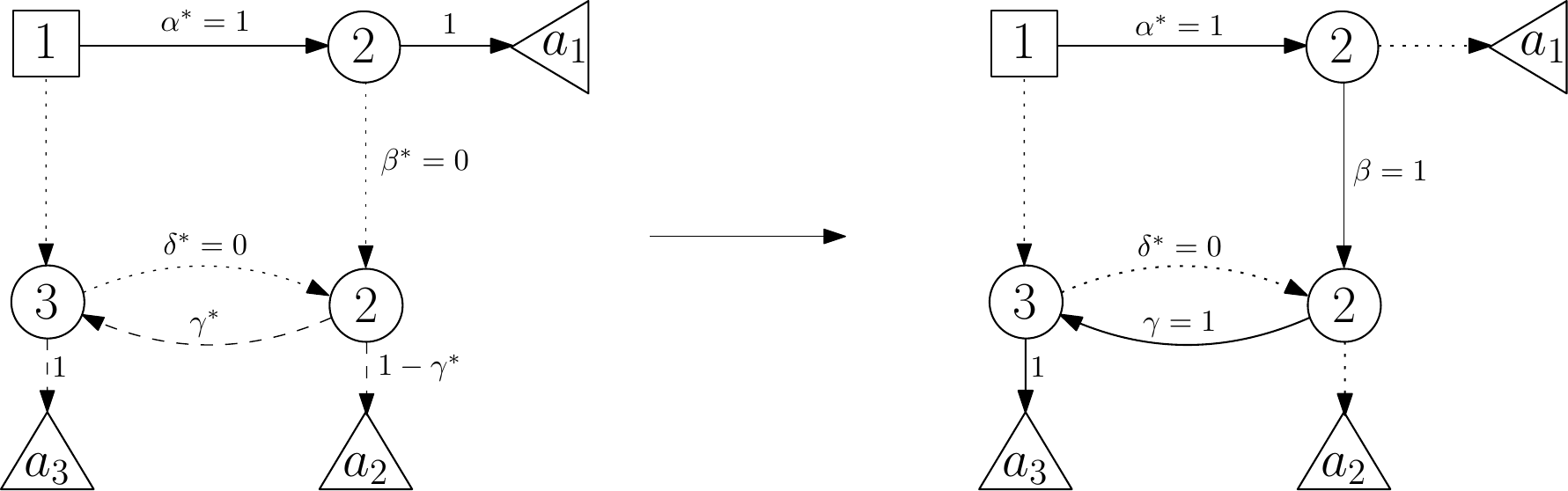}
\end{center}

It follows that $\delta^*>0$. Suppose that $\alpha^* <1$. Then $(1-\alpha^*)\delta^* > 0$ and hence $\gamma^* = 0$. Since
$1-\alpha^*+\alpha^*\beta^*\gamma^* > 0$ and $\delta^*>0$, this contradicts Observation 3A. We thus have $\alpha^* = 1$.

Next, note that we have either
$\frac{u_1(a_1)}{u_1(c)}<\delta^*$ or
$\frac{u_2(a_1)}{u_2(a_3)}<1-\delta^*$, since otherwise
we would have
$\frac{u_1(a_1)}{u_1(c)}+\frac{u_2(a_1)}{u_2(a_3)}\ge 1$, contradicting
inequality~\eqref{assumption3}.

Suppose first that $u_1(a_1)<\delta^*u_1(c)$. Then, player $1$ can strictly improve his/her expected payoff from $\phi^1(1,0,\gamma^*,\delta^*) = u_1(a_1)$ to
$$\phi^1(0,0,\gamma^*,\delta^*) = \delta^*(1-\gamma^*)u_1(a_2)+\gamma^*\delta^* u_1(c)>
\delta^*((1-\gamma^*)u_1(c)+\gamma^*u_1(c)) = \delta^*u_1(c)>u_1(a_1)\,,$$
by changing his/her strategy to $\alpha = 0$.

\begin{center}
\includegraphics[width=0.7\textwidth]{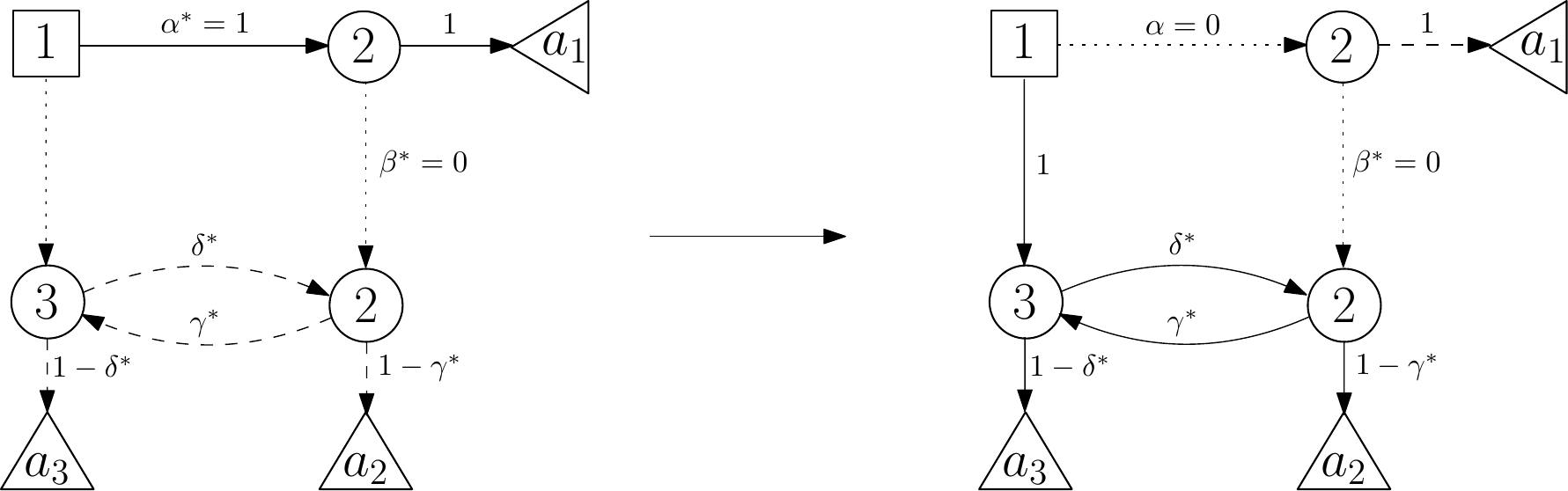}
\end{center}

This contradicts the fact that $s^*$ is an NE.

Suppose now that $u_2(a_1)<(1-\delta^*)u_2(a_3)$. Then, player $2$ can strictly improve his/her expected payoff from $\phi^2(1,0,\gamma^*,\delta^*) = u_2(a_1)$ to $\phi^2(1,1,1,\delta^*) = (1-\delta^*)u_2(a_3)$ by changing his/her strategy to $(\beta,\gamma) = (1,1)$.

\begin{center}
\includegraphics[width=0.7\textwidth]{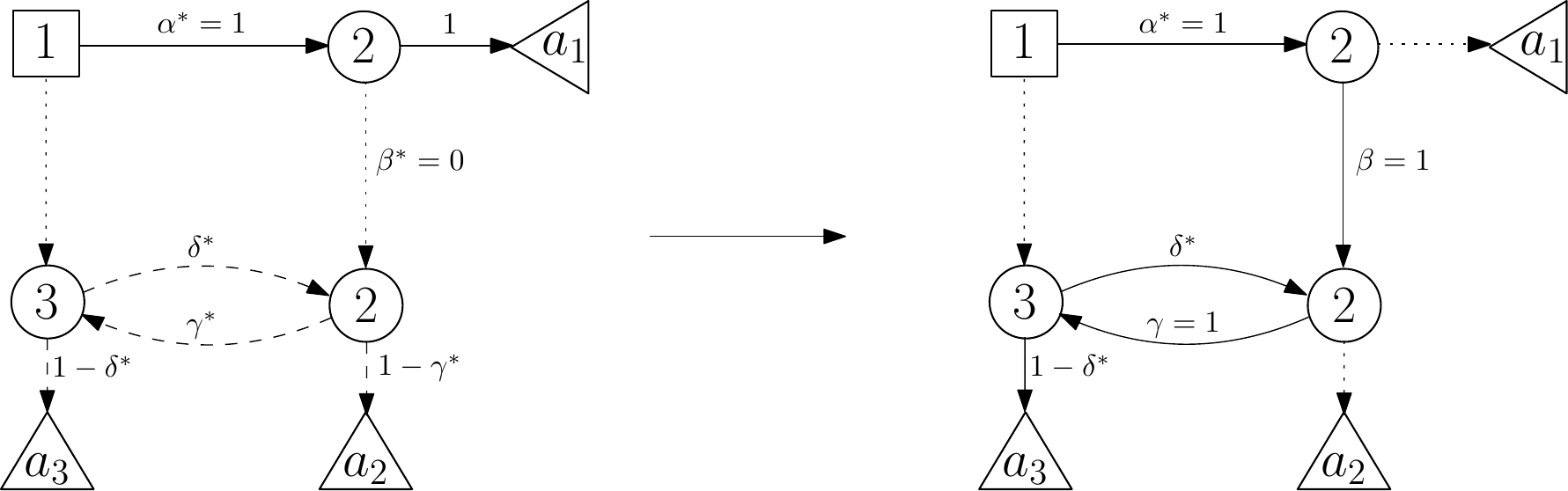}
\end{center}

This contradicts the fact that $s^*$ is an NE and completes the proof of Claim~\ref{claim6}.
\hfill$\blacktriangle$

\medskip
\begin{claim}\label{claim7}
$0 < \beta^* < 1$ and $\delta^*\in \{0,1\}$.
\end{claim}

\noindent
{\it Proof.}
The fact that $\delta^*\in \{0,1\}$ follows from Claims~\ref{claim5} and~\ref{claim6}; it thus
suffices to show that $0<\beta^*<1$. Suppose this is not the case. Then $\beta^* = 1$ by Claim~\ref{claim6}.
Observation 1 implies that either $\delta^* = 1$, $\gamma^* = 1$, or $\alpha^* = 1$.

Suppose that $\delta^* = 1$. Then $\alpha^* = 0$ by Observation 2A.
Furthermore, Observation 2B implies that $\gamma^* = 0$.
Therefore $s^*= (0,1,0,1)$ is an NE of $(G,D,u,v_0)$
in pure strategies, contradicting Proposition~\ref{prop:no-NE-in-pure-strategies}.
We conclude that $\delta^*=0$, which implies, in particular, that $\gamma^* = 1$ or $\alpha^* = 1$.

Since $\beta^* = 1$, we have $1-\alpha^*+\alpha^*\beta^* = 1$ and
Observation 3B implies that $\gamma^* < 1$. Therefore, $\alpha^* = 1$.
But now, player $2$ can strictly improve his/her expected payoff from
$\phi^2(1,1,\gamma^*,0) = (1-\gamma^*)u_2(a_2)+\gamma^*u_2(a_3)$
to $\phi^2(1,1,1,0) = u_2(a_3)$
by changing his/her strategy to $(\beta,\gamma) = (1,1)$.

\begin{center}
\includegraphics[width=0.7\textwidth]{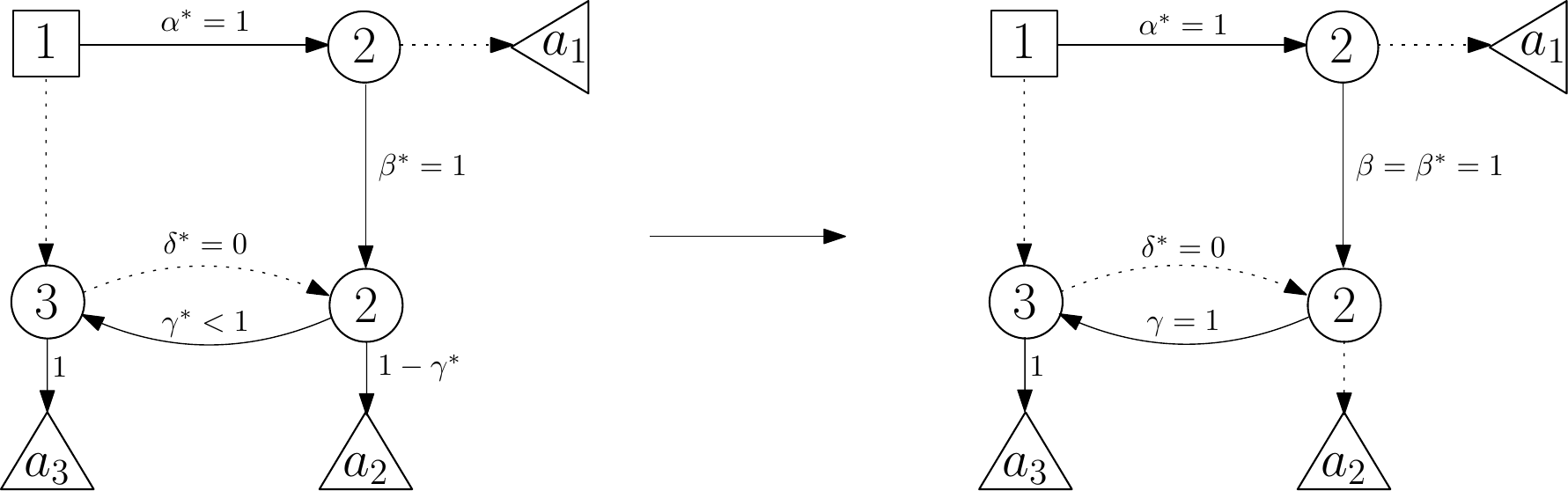}
\end{center}

This contradicts the fact that $s^*$ is an NE and completes the proof of Claim~\ref{claim7}.
\hfill$\blacktriangle$

\medskip
\begin{claim}\label{claim8}
$\delta^* = 0$.
\end{claim}

\noindent
{\it Proof.}
Suppose that $\delta^*>0$. Then $\delta^* = 1$ by Claim~\ref{claim7}.
By the same claim, $\beta^* > 0$, which implies that
$\alpha^*\beta^*+(1-\alpha^*)\delta^*> 0$,
hence Observation 2B yields $\gamma^* = 0$.
Next, Observation 3A implies that $\alpha^* = 1$.
But this contradicts Observation 2A. \hfill$\blacktriangle$

\medskip
We now have everything ready to derive the final contradiction, which will complete the proof of Theorem~\ref{thm:a-priori}.
We have $0<\beta^*<1$ and $\delta^* = 0$ by Claims~\ref{claim7} and~\ref{claim8}.
From Observation 1, we infer $\alpha^* = 1$.
This in turn implies, using Observation 2A, that $\gamma^*> 0$.
But now, player $2$ can strictly improve his/her expected payoff from
$\phi^2(1,\beta^*,\gamma^*,0) = (1-\beta^*)u_2(a_1)+\beta^*(1-\gamma^*)u_2(a_2)+\beta^*\gamma^*u_2(a_3)$
to $\phi^2(1,1,1,0) = u_2(a_3)$
by changing his/her strategy to $(\beta,\gamma) = (1,1)$.

\begin{center}
\includegraphics[width=0.7\textwidth]{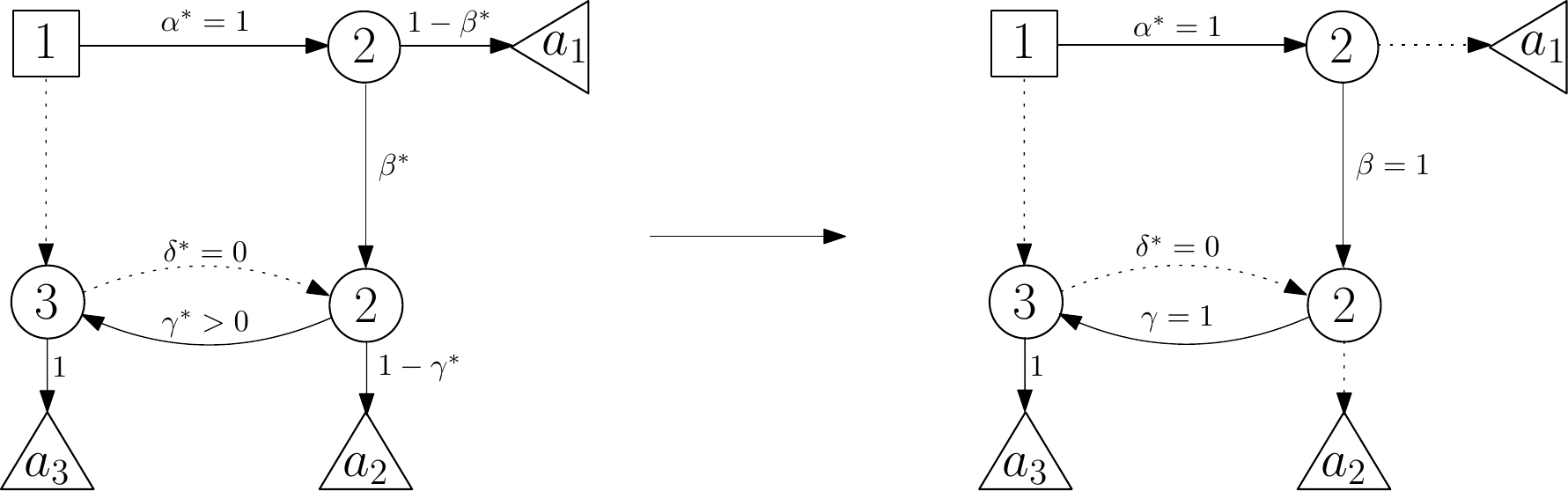}
\end{center}

This contradicts the fact that $s^*$ is an NE and completes the proof.

\subsection*{Acknowledgement}

The authors are grateful to the anonymous referees for helpful remarks.

\end{document}